\documentclass[10pt,journal,compsoc]{IEEEtran}

\usepackage{paralist,amsbsy}
\usepackage{stmaryrd,mathrsfs,url}
\usepackage{cuted,lipsum}
\usepackage[utf8]{inputenc}
\usepackage[T1]{fontenc}
\usepackage{amssymb,amsmath,amsfonts}
\usepackage{mathtools}
\usepackage{algorithm,algorithmic}
\usepackage{cite}
\usepackage{float}
\usepackage{graphics}
\usepackage{subcaption}
\usepackage[usenames,dvipsnames]{xcolor}
\usepackage{epsfig, hyperref}
\usepackage{verbatim}
\usepackage{url,changebar,bm,dsfont}
\usepackage[T1]{fontenc}
\usepackage{stfloats}
\usepackage{cleveref}

\usepackage{soul}

\usepackage{hyperref}
\hypersetup{
    colorlinks=true,
    linkcolor=black,
    filecolor=magenta,      
    urlcolor=blue,
    citecolor=Orange,
}

\usepackage[font=small]{caption}

\floatstyle{ruled}
\newfloat{model}{H}{mod}
\floatname{model}{\footnotesize Model}
\newfloat{notatio}{H}{not}
\floatname{notatio}{\footnotesize Notation}

\usepackage{url,changebar,bm,xspace,dsfont}
\let\mathbb=\mathds %

\let\mathbb=\mathds %

\newtheorem{remark}{\bfseries Remark}
\newtheorem{theorem}{\bfseries Theorem}
\newtheorem{lemma}{\bfseries Lemma}
\newtheorem{proof}{\bfseries Proof}
\newtheorem{assumption}{\bfseries Assumption}

\newenvironment{list5}{
  \begin{list}{$\bullet$}{
      \setlength{\itemsep}{0.05cm}
      \setlength{\labelsep}{0.2cm}
      \setlength{\labelwidth}{0.3cm}
      \setlength{\parsep}{0in}
      \setlength{\parskip}{0in}
      \setlength{\topsep}{0in}
      \setlength{\partopsep}{0in}
      \setlength{\leftmargin}{0.17in}}}
      {\end{list}}

\newcommand{\CPU}{\texttt{CPU}\xspace}

\begin{document}
\title{CPU Scheduling in Data Centers Using Asynchronous Finite-Time Distributed Coordination Mechanisms}

\author{Andreas Grammenos, Themistoklis Charalambous,~\IEEEmembership{Senior Member,~IEEE}, and Evangelia Kalyvianaki

\IEEEcompsocitemizethanks{
    \IEEEcompsocthanksitem A. Grammenos is with the Department of Computer Science and Technology, University of Cambridge, Cambridge, and the Alan Turing Institute, London, UK. Email: {\tt ag926@cl.cam.ac.uk}.
    \IEEEcompsocthanksitem T. Charalambous is with the Department Electrical and Computer Engineering, School of Engineering, University of Cyprus. He is also a Visiting Professor with the Department of Electrical Engineering and Automation, School of Electrical Engineering, Aalto University, Espoo, Finland. Email: {\tt themistoklis.charalambous@aalto.fi}.
    \IEEEcompsocthanksitem E. Kalyvianaki is with the Department of Computer Science and Technology, University of Cambridge, Cambridge, UK. Email: {\tt ek264@cl.cam.ac.uk}.
}%
}

\markboth{IEEE Transaction on Network Science and Engineering}%
{Grammenos \MakeLowercase{\textit{et al.}}: Distributed CPU Scheduling in Data Centers Using Asynchronous Finite-Time Distributed Coordination Mechanisms}

\IEEEtitleabstractindextext{
\begin{abstract}
We propose an asynchronous iterative scheme that allows a set of interconnected nodes to distributively reach an agreement within a pre-specified bound in a finite number of steps. 
While this scheme could be adopted in a wide variety of applications, we discuss it within the context of task scheduling for data centers. 
In this context, the algorithm is guaranteed to \emph{approximately} converge to the optimal scheduling plan, given the available resources, in a finite number of steps.
Furthermore, by being asynchronous, the proposed scheme is able to take into account the uncertainty that can be introduced from straggler nodes or communication issues in the form of latency variability while still converging to the target objective.
In addition, by using extensive empirical evaluation through simulations we show that the proposed method exhibits state-of-the-art performance.
\end{abstract}
}
\maketitle

\IEEEraisesectionheading{\section{Introduction}\label{sec:intro}}

\IEEEPARstart{C}{loud computing} provides software and hardware resources on demand via the Internet and has become the predominant model for application deployment. 
The backbone of modern Cloud infrastructure consists of a network of data centers, each equipped with thousands of server machines, running diverse application workloads, supporting uncoordinated and heterogeneous users and their applications~\cite{barroso_datacenter_2018}. 
Data center resource management is the fundamental task of allocating resources (e.g., \CPU, memory, network bandwidth, and disk space) to workloads such that their performance objectives are satisfied and the overall data center utilization is kept high~\cite{cortez_resource_2017}. 
Notably, even slight deviations from the desired objectives can have substantial detrimental effects with millions of dollars in revenue potentially lost~\cite{amvrosiadis_diversity_2018}. 
Therefore, scheduling in data centers is the most fundamental operation responsible for allocating resources to workloads while satisfying their performance requirements~\cite{boutin_apollo_2014}.
In doing so, scheduling aims to find the best placement of jobs within the available compute nodes that maximizes the overall utilization of resources and which can ultimately lead to a massive reduction in operational and capital costs.

More formally, scheduling can be viewed as an \emph{optimization problem} in which workloads are allocated to server machines such that a performance goal is optimized while all constraints are satisfied~\cite{moritz_ray_2018,verma_large-scale_2015}. 
In this paper, we focus on \emph{minimizing the sum of \CPU utilization across servers}. 
In other words, the workload should be shared {proportionally} across servers {based on their hardware}, such that they all use the minimum percentage of their capacity and essentially the total workload at each server node is balanced {and proportional to its available resources}. 
The main reason for this formulation is to avoid overloading specific servers and so to efficiently serve workloads.
Solving a scheduling optimization problem in such a large-scale system {in a centralized fashion} is challenging due to the size of the network and the dynamic nature of resource requirements of incoming and existing workloads. 
{
In general, centralized approaches in multi-component systems require the collection of measurements or other information to a central location (at possibly high communication and computational cost), the computation of quantities of interest at this central location, and then the dissemination of these quantities to (a subset of) the components. 
This approach, as it is the case in Clouds, is often inefficient. 
This is because centralised approaches focus the entire load towards a single node.
This can not only be a point of failure, but also create congestion in the network causing often causing delays and spikes in response times~\cite{boutin_apollo_2014,hindman_mesos_2011,ousterhout_sparrow_2013}.
Cooperative distributed coordination algorithms have therefore received tremendous attention, especially during the last two decades. 
Several diverse research communities (e.g., biology, physics, control, communication and computer science) have made several contributions that have resulted in many recent advances in so called consensus-based approaches (see, for example, \cite{olfati-saber_consensus_2007}) and in distributed computation of functions of geographically dispersed data, also known as in-network computation (see, for example, \cite{giridhar_toward_2006} and references therein).} %

{
Classical approaches in distributed coordination algorithms typically assume timely and reliable exchange of information between neighboring components of a given multi-component system. 
These assumptions are not necessarily valid in practical settings due to varying delays that might affect transmissions at different times, as well as possible changes in the underlying interconnection topology (e.g., due to unexpected cluster changes as nodes randomly fail and/or abnormal runtime behaviors due to software or configuration faults and resource contention)~\cite{barroso_datacenter_2018, misra_managing_2019}.
}

{In this work, we propose a distributed coordination protocol to overcome these limitations.}
To this end, we posit a novel scheme that takes in account these potential latency variations in the form of explicit delays in the communication links during planning, while still remaining asynchronous in its operation and we guarantee that it will converge in finite-time.

\subsection{Contributions}
For the context of this work, we formulate the \CPU scheduling as a distributed optimization problem and solve it using distributed coordination mechanisms. 
More concretely, the contributions of the paper are as follows.
\begin{list5}
    \item {
        First, using existing theory from optimization, we provide the closed form solution, which requires the knowledge of global parameters, such as, the total capacity of the network and the total incoming workload.
    }
    \item {
        Second, it is shown that the problem can be solved in a distributed fashion.
    }
    \item {
        When the updates of the nodes are synchronous, we adopt a mechanism which uses a well-known consensus algorithm (namely, ratio consensus) proposed in~\cite{cady_finite-time_2015}, with which an \emph{approximate} solution is reached in a finite number of steps.
    }
    \item {
        When the updates of the nodes are asynchronous, we adopt a mechanism, {of} similar {flavor} to the one proposed in~\cite{prakash_distributed_2020}, in which finite-time average consensus is achieved in the presence of bounded time-varying delays. 
        More specifically, our proposed algorithm allows the nodes to distributively compute the optimal value to within an error bound in a finite number of steps. 
        The methodology builds upon \emph{(i)} robustified ratio consensus~\cite{hadjicostis_asynchronous_2011,hadjicostis_average_2014}, a distributed iterative algorithm in which each node maintains two state variables where the ratio of the states converges asymptotically to a constant that is equal for all the nodes, and \emph{(ii)} asynchronous $\max-$consensus algorithm~\cite{giannini_convergence_2013}.
        
    }
    \item {
        Finally, numerical examples and evaluations show the efficacy of the proposed solutions.
    }
\end{list5}

The main benefit of our approach is that the global optimization problem is decomposed into local objectives and the problem is then solved in a distributed manner via our proposed distributed coordination mechanisms, which provide a way for the nodes to terminate iterations simultaneously, while ensuring at the same time that the worst-case error lies within the pre-specified bound. 
These properties make these mechanisms suitable for applications in which (repeated) optimization problems have to be solved fast and in a finite number of steps. 
{Moreover, contrary to methods such as ADMM our scheme requires significantly less resources for its computation to reach similar objectives as can be seen from the results put forth in recent studies~\cite{chang_asynchronous_2016-1, jiang_asynchronous_2021}.
This property can be particularly useful as most scheduling operations assume minimal processing latency to reach a solution for the optimal placement of tasks. }
{To the best of our knowledge, this is the first algorithm with finite-time termination guarantees that can handle delays and provide asynchronous consensus.}

\subsection{Related Work}

\subsubsection{Data Center Scheduling}

{Centralized data center schedulers such as~\cite{isard_quincy_2009, gog_firmament_2016, mao_learning_2019, tumanov_tetrisched_2016, grandl_multi-resource_2014,2021:WangPSO}, as well as the well known centralized  approaches of Google's Borg and Kubernetes~\cite{burns_borg_2016} and Microsoft's Resource Central~\cite{cortez_resource_2017}, are able to provide optimized scheduling decisions under specific constraints and goals. 
More recently, there have been some centralized schedulers that tackle utilisation optimization focusing on energy efficiency~\cite{yuan_biobjective_2020, yuan_fine-grained_2020, yuan_geography-aware_2020}.
However, they require continuous transferring of resource information at the centralized scheduler which increases data center network traffic. 
Furthermore, centralized schedulers typically lack of large-scale scalability and they can be a single point of failure.}
In contrast, our distributed approach requires each node to send its estimated utilization to its out-neighbors only reducing therefore the total amount of information sent and uses the most up-to-date resource estimates for more accurate scheduling.

Popular decentralized schedulers such as~\cite{hindman_mesos_2011, ousterhout_sparrow_2013, schwarzkopf_omega_2013, boutin_apollo_2014,vavilapalli_apache_2013}, as well as the most recent primary autoscaler that Google uses on
its internal cloud \cite{rzadca_autopilot_2020}, aim to tackle data center scalability by allowing different scheduling decisions to occur in parallel by multiple schedulers. 
Such approaches span a wide spectrum of schedulers' coordination---from schedulers operating independently from each other (e.g.,~\cite{ousterhout_sparrow_2013}) to schedulers sharing some global resource information (e.g.,~\cite{schwarzkopf_omega_2013,boutin_apollo_2014})---and they also differ in the way they detect and resolve conflicts in the allocation of shared resources. 
We remark that, while these solutions exhibit good empirical performance, lack formal guarantees and largely work by using heuristics~\cite{gan_open-source_2019}.
Unfortunately, this can be problematic when volatile or unpredictable workloads are encountered~\cite{ousterhout_sparrow_2013}. 
Additionally, state sharing can be problematic in case of delays as such schedulers attempt to globally infer the state of the cluster and normally are not able to tolerate delays.
This in turn can lead to suboptimal performance~\cite{jeon_analysis_2019}.
In contrast in our distributed approach all nodes/schedulers coordinate \textit{asynchronously} by design to find optimal allocations at scheduling time without facing any conflicts.

Multi-resource allocation of tasks to data center nodes is known to be a APX-Hard~\cite{mao_learning_2019}. 
Most scheduling approaches employ heuristics to solve the problem in reasonable timeframes~\cite{mao_learning_2019, grandl_graphene_2016, tumanov_tetrisched_2016, boutin_apollo_2014}. 
Fewer approaches tackle the problem using appropriate centralized solvers (e.g., IBM's CPLEX in~\cite{tumanov_tetrisched_2016}) albeit for small problem sizes compared to today's data center sizes of thousands of nodes. 
Such approaches highly depend on the compute and memory capacity of the centralized solver to handle hundreds of thousands of constraints typically present in such problem formulations. 

Our approach is to formulate the problem of CPU task scheduling in data centers as a distributed optimization one to solve it using distributed coordination mechanisms. 
An approximate (not accurate) solution can be computed in a finite number of steps and is guaranteed to complete while exhibiting graceful scaling. 
These properties enable its application to data center sized scheduling problems containing even tens of thousands of participating nodes.

\subsubsection{Distributed finite-time average consensus}

{%
A distributed system or network consists of a set of components (nodes) that can share information via connection links (edges), forming a directed interconnection topology (directed graph). 
In general, the objective of a consensus problem is to have all agents agree upon a certain ({\em a priori} unknown) quantity of interest that is typically a function of some values that the nodes initially posses. 
When the agents (asymptotically) reach an agreement to the same value, we say that the distributed system (asymptotically) reaches consensus. 
Such problems include network coordination problems involving self-organization, formation of patterns, parallel processing, and distributed optimization. 
The problem of convergence of discrete-time consensus algorithms was initially targeted by Tsitsikis \emph{et al.} \cite{tsitsiklis_distributed_1986} and subsequently by many other researchers (see, for example, \cite{olfati-saber_consensus_2004, ren_consensus_2005, moreau_stability_2005, angeli_stability_2006, bliman_average_2008, cai_average_2011}). %
Convergence of consensus algorithms can usually be established under relatively weak requirements. 
Common challenges include the handling of node failures, transmission delays on the transfer of data between agents, packet losses in wireless communication networks, and inaccurate sensor measurements. 
As a result, it is imperative to address agreement problems that consider networks of dynamical agents, possibly with directed information flow, under delays and/or changing topologies.
One of the most well known consensus problems is the so-called \emph{average consensus} problem in which agents aim to reach the average of their initial values (see, for example, \cite{jadbabaie_coordination_2003, xiao_fast_2004}). %
} 

This work is based on \emph{synchronous} and \emph{asynchronous} finite-time average consensus algorithms. 
There have been several works on synchronous finite-time average consensus algorithms due to their use \emph{i)} in  resource-constrained applications (such as wireless sensor networks) since they save energy and computational resources, and \emph{ii)} in applications in which the result of the consensus algorithm is used in real-time to perform other subsequent tasks (such as smart energy networks). 
Nevertheless, there have not been any works for the asynchronous case when consensus is achieved in a finite number of steps. 

The model of asynchrony considered herein allows for heterogeneous, but bounded computation and communication delays, thus quantifying the degree of asynchrony by a bound on the time-delays. 
It is highlighted that the nodes are not required to know the bound for the execution of the algorithm. 
Finite-time average consensus in the presence of delays in directed graphs has been studied mainly by~\cite{charalambous_distributed_2015} for exact average consensus and more recently by~\cite{khatana_gradient-consensus_2020} for approximate average consensus. 
Moreover, the bound provided has linear dependency to the maximum delay within the network multiplied with its diameter. 
This is a powerful result, as not only allows its applicability in traditional data centers where consensus can be achieved quickly but also in delay tolerant networks. 
This particular category includes numerous types of networks with some notable examples being collaborative autonomous agents, mobile phones, IoT clusters, and others.

\subsection{Organization}

The remainder of the paper is organized as follows. 
In Section~\ref{sec:system}, we give the necessary notation and describe the model of the system. 
In Section~\ref{sec:preliminaries}, we provide the necessary background knowledge needed for the development of our results. 
In Section~\ref{sec:mainresults}, we first provide the problem under consideration and then we modify it so that it is formulated as a distributed coordination. 
Next, in Sections~\ref{sec:distributed_synchronous} and~\ref{sec:distributed_asynchronous} we propose a synchronous and an asynchronous finite-time distributed algorithm, respectively, that solve the problem approximately. 
In Section~\ref{sec:simulations}, we demonstrate the efficacy of our proposed algorithms. 
In Section~\ref{sec:discussion}, we provide a quantitative discussion of the contributions herein and discuss our findings.
In Section~\ref{sec:conclusions} we draw conclusions and discuss possible directions for future work.

\section{Notation and System Model}
\label{sec:system}

\subsection{Notational Conventions}
\label{subsec:notation}

The set of real (integer) numbers is denoted by $\mathds{R}$ ($\mathds{Z}$) and the set of non-negative real (integer) numbers is denoted by $\mathds{R}_{+}$ ($\mathds{Z}_{+}$). 
Vectors are denoted by small letters whereas matrices are denoted by capital letters.  $A^T$ denotes the transpose of matrix $A$. The $i^{\text{th}}$ component of a vector $x$ is denoted by $x_i$. 
For $A\in \mathbb{R}^{n\times n}$, $a_{ij}$ denotes the entry in row $i$ and column $j$. 

In multi-component systems with fixed communication links (edges), the exchange of information between components (nodes) can be conveniently captured by a graph $\mathcal{G}(\mathcal{V}, \mathcal{E})$ of order $n$ $(n \geq 2)$, where $\mathcal{V} = \{v_1,v_2,\ldots,v_n\}$ is the set of nodes and $\mathcal{E} \subseteq \mathcal{V} \times \mathcal{V}$ is the set of edges. 
An edge from node $v_{i}$ to node $v_{j}$ is denoted by $\varepsilon_{ji} = (v_{j}, v_{i})\in \mathcal{E}$ and represents a communication link that allows node $v_{j}$ to receive information from node $v_{i}$. 
A graph is said to be undirected if and only if $\varepsilon_{ji} \in \mathcal{E}$ implies $\varepsilon_{ij}  \in \mathcal{E}$. 
A digraph is called connected if there exists a path from each vertex $v_{i}$ of the graph to each vertex $v_{j}$ ($v_{j} \neq v_{i}$). 
The diameter $D$ of a graph is the longest shortest path between any two nodes in the network.

In \emph{digraphs}, nodes that can transmit information to node $v_{j}$ directly are said to be in-neighbors of node $v_{j}$ and belong to the set $\mathcal{N}^{-}_j=\{ v_{i} \in \mathcal{V} \; | \; \varepsilon_{ji} \in \mathcal{E} \}$. 
The cardinality of $\mathcal{N}^{-}_{j}$, is called the \emph{in-degree} of $v_{j}$ and is denoted by $\mathcal{D}^{-}_{j}=\left| \mathcal{N}^{-}_j \right|$. 
The nodes that receive information from node $v_{j}$ belong to the set of out-neighbors of node $v_{j}$, denoted by $\mathcal{N}^{+}_j=\{ v_l \in \mathcal{V} \; | \; \varepsilon_{lj} \in \mathcal{E} \}$. 
The cardinality of $\mathcal{N}^{+}_j$, is called the \emph{out-degree} of $v_{j}$ and is denoted by $\mathcal{D}^{+}_{j}= \left| \mathcal{N}^{+}_j \right|$.

{In the type of algorithms we will consider, we will associate a positive weight $p_{ji}$ for each edge $\varepsilon_{ji} \in \mathcal{E} \cup \{ (v_j, v_j) \; | \: v_j \in \mathcal{V} \}$. The nonnegative matrix $P = [p_{ji} ] \in \mathbb{R}_{+}^{n\times n}$ (with $p_{ji}$ as the entry at its $j$th row, $i$th column position) is a weighted adjacency matrix (also referred to as weight matrix) that has zero entries at locations that do not correspond to directed edges (or self-edges) in the graph. In other words, apart from the main diagonal, the zero-nonzero structure of the adjacency matrix $P$ matches exactly the given set of links in the graph. }

\subsection{System Model}
\label{subsec:model}

In our setup, we assume a set $\mathcal{V}$ of server compute nodes, denoted by $v_{i}\in\mathcal{V}$, which also operate as resource schedulers; this is a frequent occurrence in modern data-centers. 
All participating schedulers are interconnected with bidirectional communication links and, thus, the network topology forms a connected undirected graph.

A job is defined as a group of tasks and $\mathcal{J}$ as the set of all jobs to be scheduled. 
Each job $b_{j} \in \mathcal{J}$, $j\in\{1,\ldots, |\mathcal{J}| \}$ requires $\rho_{j}$ cycles to be executed and their individual estimated cost is assumed to be known before the optimization starts.
The time horizon of the optimization (denoted by $T_{h}$) is defined as the time period for which the optimization is considering the jobs to be running on the server nodes, before the next optimization decides the next allocation of resources. 
Hence, the \CPU capacity of each node, considered during the optimization, is computed as
\begin{align}
    \pi_i^{\max} \coloneqq c_i T_h, 
\end{align}
where $c_i$ is the sum of all clock rate frequencies of all processing cores of node $v_{i}$ given in cycles/second. 
The \CPU availability for node $v_{i}$ at optimization step $m$ (i.e., at time $mT_{h}$) is given by
\begin{align}
    \pi_i^{\mathrm{avail}}[m] \coloneqq \pi_i^{\max} - u_{i}[m], 
\end{align}
where $u_{i}[m]$ is the number of unavailable/occupied cycles due to predicted or known utilization from already running tasks on the server over the time horizon $T_{h}$ at step $m$.

\begin{assumption}
    Since the time horizon $T_h$ is a parameter chosen, we assume that the time horizon is chosen such that the total amount of resources demanded at a specific optimization step $m$, denoted by $\rho[m] \coloneqq \sum_{b_j[m] \in \mathcal{J}[m]} \rho_{j}[m]$, is smaller than the total capacity of the network available, given by $\pi^{\mathrm{avail}}[m] \coloneqq \sum_{v_{i}\in \mathcal{V}} \pi_i^{\mathrm{avail}}[m]$, i.e., $\rho[m] \leq \pi^{\mathrm{avail}}[m]$. 
\end{assumption}

This assumption indicates that there is no more demand than the available resources. 
In case this assumption is violated, the solution will be that all resources are being used and some workloads will not be scheduled, due to lack of resources.
The workloads selected to be discarded are arbitrary and the purging does not adhere to any particular priority policy; the jobs are scheduled on a first-come, first-scheduled basis.
A more sophisticated priority mechanism could be deployed whose task would be to allocate a subset of the requests only, based on some optimization problem (taking into account deadlines, etc). However, since the prioritization problem is out of the scope of this paper, this problem will be addressed separately.

In particular, in the context of this work, we focus on CPU utilization as it is one of the most important and precious resource for many workloads. 
We note that while communication costs are also important, as data center networking becomes faster, we believe that CPU remains the most important resource. 
Notably, forecasting resource demands can be challenging and costly~\cite{burns_borg_2016,schwarzkopf_omega_2013}, without necessarily providing the expected gains.
Further, there are certain types of commonly encountered jobs (e.g., recurring batch processing workloads) that have known a-priory demands. 
In fact, such workloads are frequently encountered in enterprise environments as such current data center scheduler's (e.g., Google's Kubernetes~\cite{bernstein_containers_2014}) operate on known resource demands.

\section{Preliminaries}\label{sec:preliminaries}

\subsection{Average Consensus}
\label{subsec:coordination}

{In a synchronous setting, each node $v_j$ updates and sends its information to its out-neighbors (and also receives similar information from its in-neighbors) at discrete times $t(0), t(1), t(2), \ldots$. 
We index nodes' information states and any other information at time $t(k)$ by $k$. We use $x_j[k]\in \mathbb{R}$ to denote the information state of node $v_{j}$ at time $t_k$. 
In our setup,} each node $v_{j}$  updates and sends its information regarding its input workload $\ell_j$ ($\ell_j$ is the summation of workloads at node $v_{j}$), estimated needed utilization for other tasks $u_{j}$, and capacity $\pi_j^{\max}$ to its out-neighbors .

At each step, node $v_{j}$ updates its information state $x_{j}[k]$ by combining the available information received by its neighbors $x_{i}[k]$ ($v_{i}\in \mathcal{N}^{-}_j$) using a weighted linear combination, i.e., 
\begin{align}\label{eq:1_1}
    x_{j}[k+1] =p_{jj}[k] x_{j}[k] + \sum_{v_{i} \in \mathcal{N}^{-}_j}  p_{ji}[k]  x_{i}[k] \; , \ \  k\geq 0 \;,
\end{align}
where  $x_{j}[0] \in \mathbb{R}$ is the initial state of node $v_{j}$. 
The positive weights $p_{ji}[k]$ capture the weight of the information inflow from node $v_{i}$ to node $v_{j}$ at time $k$ (note that unspecified weights in $P$ correspond to pairs of nodes $(v_{j},v_{i})$ that are not connected and are set (without loss of generality) to zero, i.e., $p_{ji}[k]=0$, $\forall \varepsilon_{ji} \notin \mathcal{E}$). 
If we let $x[k]=(x_1[k] \ \ x_2[k] \ \ \ldots  \ \  x_n[k] )^T$ and $P[k] = [p_{ji}[k] ] \in \mathbb{R}_{+}^{n\times n}$, then~\eqref{eq:1_1} can be written in matrix form as
\begin{align}\label{eq:2_1}
    x[k+1] =P[k] x[k] ,
\end{align}
\noindent where $x[0] = (x_1[0] \ \ x_2[0] \ \ \ldots  \ \  x_n[0] )^T \equiv x_0$. 
In this work, we consider a static network; as a result, the graph remains invariant. 
In this case, the weights can be chosen to be constant for all times $k$ (i.e., $p_{ji}[k]=p_{ji} \ \forall k$), and equation~\eqref{eq:2_1} can be expressed as \begin{align}\label{eq:2_2}
    x[k+1] =P x[k] .
\end{align}

{%
We say that the nodes asymptotically reach average consensus if $\lim_{k \rightarrow \infty} x_j[k] = \frac{\sum_i x_i[0]}{n}$ for all $v_j \in \mathcal{V}$. 
The necessary and sufficient conditions for \eqref{eq:2_2} to reach average consensus are the following \cite{xiao_fast_2004}: (a) $P$ has a simple eigenvalue  $\lambda_i(P)=1$ with left eigenvector $\mathbb{1}^T$ and right eigenvector $\mathbb{1}$, and (b) all other eigenvalues of P ($\lambda_j(P), j\neq i$)  have magnitude less than 1 ($|\lambda_j(P)|<1$). 
If $P\geq 0$ (as in our case), the necessary and sufficient condition is that $P$ be a primitive doubly stochastic matrix. 
}

\subsection{Ratio consensus} 
\label{subsec:RATIO}

Dominguez-Garc\'{i}a and Hadjicostis in~\cite{dominguez-garcia_coordination_2010}, propose an algorithm that solves the average consensus problem in a directed graph in which each node $v_{j}$ distributively sets the weights on its self-link and outgoing-links to be $p_{lj}=\frac{1}{1+\mathcal{D}_j^+}$, $\forall (v_l, v_{j})\in \mathcal{E}$, so that the resulting weight matrix $P$ is column stochastic, but not necessarily row stochastic. 
Average consensus is reached by using this weight matrix to run two iterations with appropriately chosen initial conditions. The algorithm is stated below for the specific choice of weights mentioned above (which assumes that each node knows its out-degree). 
Note, however, that {the algorithm also works for any set of weights that adhere to the graph structure and form a primitive column stochastic weight matrix}.

\begin{lemma}[\hspace{-0.001cm}\cite{dominguez-garcia_coordination_2010}]
\label{lemma_christoforos} 
    Consider a strongly connected graph $\mathcal{G}(\mathcal{V}, \mathcal{E})$. Let $y_{j}[k]$ and $z_{j}[k]$ (for all $v_{j} \in \mathcal{V}$ and $k=0,1,2,\ldots$) be the result of the iterations
    \begin{subequations}
        \begin{align}\label{eq:4}
            x_j[k+1]=p_{jj} x_j[k]+ \sum_{v_{i} \in \mathcal{N}^{-}_j} p_{ji} x_i[k] \; , \\  
            y_{j}[k+1]=p_{jj} y_{j}[k]+ \sum_{v_{i} \in \mathcal{N}^{-}_j} p_{ji} y_{i}[k] \; ,
        \end{align}
    \end{subequations}
    where $p_{lj} = \frac{1}{1 + \mathcal{D}_j^+}$ for $v_l \in \mathcal{N}_j^+$ (zeros otherwise), and the initial conditions are $x[0]=(x_0(1) \ \ x_0(2) \ \ldots \ x_0(|\mathcal{V}|))^T\equiv x_0$ and  $y[0]=\mathbb{1}$. 
    Then, the solution to the average consensus problem can be asymptotically obtained as 
    $$
    \displaystyle \lim_{k\rightarrow \infty} \mu_{j}[k]=\frac{\sum_{v_\ell\in \mathcal{V}} x_0(\ell)}{|\mathcal{V}|} \; , \quad \forall v_{j} \in \mathcal{V} \; ,
    $$
    where $\displaystyle \mu_{j}[k]={x_j[k]}/{y_{j}[k]} \; $.
\end{lemma}
\subsection{Synchronous \texorpdfstring{$\max-$consensus}{max-consensus}} 
\label{subsec:max}

{It is desired that each node $v_{j}\in \mathcal{V}$ of a network reaches consensus on the maximum value of the initial states/measurements under the assumption that all the nodes have a single real-valued state that they update based on local received states. Each node should reach the value $x_{\max} = \max_{v_{j}\in \mathcal{V}}x_j[0] $.} 
The $\max-$ consensus algorithm is a simple algorithm for computing the maximum value {(of, e.g., initial measurements in a sensor network)} in a distributed fashion~\cite{cortes_distributed_2008}. 
When the updates are synchronous, 
{in the absence of communication noise (as it is the case in this work), $\max-$ consensus can be done by having each node $v_{j}\in \mathcal{V}$ update the state value with the largest received value in every iteration;}
the update rule is as follows:
\begin{align}
x_j[k+1] = \max_{v_{i}\in \mathcal{N}_j^{-} \cup \{v_{j}\}}\{ x_i[k] \}.
\end{align}
It has been shown (see, e.g.,~\cite[Theorem 5.4]{giannini_convergence_2013}) that this algorithm converges to the maximum value among all nodes in a finite number of steps $s$, $s \leq D$.  Similar results hold for the $\min-$consensus algorithm.

\subsection{Optimization Problem}
\label{subsec:optimization}

In a network $\mathcal{G} =(\mathcal{V},\mathcal{E})$ of $N=|\mathcal{V}|$ nodes, each node is endowed with a scalar quadratic cost function $f_i : \mathbb{R}^N \mapsto \mathbb{R}$. 
Most cases consider a quadratic cost function of the form:
\begin{align}\label{eq:fi}
f_i(z) = \frac{1}{2}\alpha_i (z-\rho_{i})^2,
\end{align}
where $\alpha_i >0$ and $\rho_{i} \in \mathbb{R}$ (in our case it is the demand in node $v_{i}$ and it is a non-negative number). 
Parameter $z$ is a function of the workload and it will be explained shortly. 
The global function $f : \mathbb{R}^N \mapsto \mathbb{R}$ is the sum of the cost function~\eqref{eq:fi} of each node $v_{i}$. 
The main goal of the nodes is to allocate the jobs in order to minimize the cost function in a distributed fashion, by communicating with their neighbors only. 
Each node is thus required to solve the following optimization problem:
\begin{align}\label{opt:1}
    z^* =  \arg\min_{z\in \mathcal{Z}} \sum_{v_{i} \in \mathcal{V}} f_i(z) , 
\end{align}
where $\mathcal{Z}$ is the set of feasible values of parameter $z$. The solution of~\eqref{opt:1} in closed form can be expressed as
\begin{align}\label{eq:closedform}
    z^* =  \frac{\sum_{v_{i} \in \mathcal{V}} \alpha_i \rho_{i}}{\sum_{v_{i} \in \mathcal{V}} \alpha_i}.
\end{align}
Note that by setting $\alpha_i =1$ for all $v_{i}\in\mathcal{V}$, the solution is the average consensus.

\section{Problem Formulation}\label{sec:mainresults}

\subsection{Problem Statement}
\label{subsec:problem-statement}

In our case, we are interested in finding a solution in which the total workload at each server node is balanced. 
This translates to having all server nodes having the same percentage of capacity being utilized during the execution of the tasks, i.e., 
\begin{align}\label{cond:balance}
    \frac{w_i^*[m] +u_{i}[m]}{\pi_i^{\max}} &= \frac{w_j^*[m] +u_{j}[m]}{\pi_j^{\max}} \\
    &= \frac{\rho[m] + u_{\mathrm{tot}}[m]}{\pi^{\max}}~\forall v_{i}, v_{j} \in \mathcal{V}, \nonumber
\end{align}
where $w_i^*[m]$ is the \emph{optimal} workload to be added to server node $v_{i}$ at optimization step $m$, $\pi^{\max} \coloneqq \sum_{v_{i}\in \mathcal{V}} \pi_i^{\max}$ and $u_{\mathrm{tot}}[m]=\sum_{v_{i}\in \mathcal{V}} u_{i}[m]$. 

The aim of this work is to find the optimal solution at every optimization step $m$ via a distributed coordination algorithm run for a finite number of steps.

\subsection{Modification of the Optimization Problem}
\label{subsec:modifiedoptimization}

To achieve the requirement set in~\eqref{cond:balance}, we modify~\eqref{eq:fi} accordingly. Let 
\begin{align}
    z[m] \coloneqq \frac{w_i[m] +u_{i}[m]}{\pi_i^{\max}}.
\end{align}
For simplicity of exposition, and since we consider a single optimization step, we drop the index $m$. 
Then, the cost function $f_i(z)$ in~\eqref{eq:fi} is given by
\begin{align}\label{eq:fiz}
    f_i(z) = \frac{1}{2}\pi_i^{\max} \left(z- \frac{\rho_{i}+u_{i}}{\pi_i^{\max}} \right)^2,
\end{align}
and the solution to problem~\eqref{opt:1} according to~\eqref{eq:closedform} is 
\begin{align}\label{eq:closedform1}
    z^* =  \frac{\sum_{v_{i} \in \mathcal{V}} \pi_i^{\max} \frac{\rho_{i}+u_{i}}{\pi_i^{\max}}}{\sum_{v_{i} \in \mathcal{V}} \pi_i^{\max}} = \frac{\rho + u_{\mathrm{tot}}}{\pi^{\max}}.
\end{align}
In other words, the nodes find the proportion of workload that each of them should have. From that each node is able to deduce the workload $w_i^*$ to receive, i.e.,
\begin{align}\label{eq:optimal_workload}
    w_i^*  = \frac{\rho + u_{\mathrm{tot}}}{\pi^{\max}} \pi_i^{\max}  -u_{i}.
\end{align}

\section{A synchronous distributed algorithm}
\label{sec:distributed_synchronous}

The solution that we are aiming for should satisfy the balance condition in~\eqref{cond:balance}. 
For each node to be able to compute the optimal workload $w_i^*$ in~\eqref{eq:optimal_workload}, the total workload $\rho$, the total estimated utilization needed for other tasks $u_{\mathrm{tot}}$, and the total capacity of the network $\pi^{\max}$ are needed. 
For solving the problem in a distributed fashion we assume the following:
\begin{assumption}\label{assumption:2}
The graph is static and strongly connected.
\end{assumption}
This assumption is, in general, valid even for large data-centers, since their topology is not expected to change for prolonged periods of time and remains mostly static, since failures are rare. Also, fault diagnostic mechanisms can be used to detect such failures and restore the connectivity of the network. Moreover, changes in the network can be handled by our algorithm, provided that the number of outgoing links can be found at each node in a distributed fashion, either because the links are bidirectional or because specific recovery schemes are deployed; see, for example, \cite{Hadjicostis:2012-CDC, 2016:Hadjicostis_TAC}.

Under Assumption~\ref{assumption:2}, running the ratio consensus algorithm~\eqref{eq:4} with initial conditions $y_{j}[0]=\ell_j+u_{j}$ and $z_ij[0]=\pi^{\max}$, we obtain
\begin{align*}
    \lim_{k\to\infty}y[k] &= \lim_{k\to\infty} P^k y[0] = c \mathbb{1}^T y[0] = %
    c(\rho+u_{\mathrm{tot}}) , \\  
    \lim_{k\to\infty}y[k] &= \lim_{k\to\infty} P^k z[0] 
    = c \mathbb{1}^T z[0] = %
    c\pi^{\max} ,
\end{align*}
where $c$ is a vector (the left eigenvector of column matrix $P$).
Therefore,
\begin{align*}
    \lim_{k\to \infty} \mu_{j}[k]= \lim_{k\to \infty} \frac{x_j[k]}{y_{j}[k]} = \frac{c_j(\rho+u_{\mathrm{tot}})}{c_j\pi^{\max} } = \frac{\rho+u_{\mathrm{tot}}}{\pi^{\max}}.
\end{align*}

\subsection{Finite-time implementation}

Since the optimization is repeated periodically, the consensus algorithm should stop way before the beginning of the next optimization cycle, since the resources should be allocated and have the tasks allocated (and process as many of them as possible) before the next bunch of tasks arrives{; see Fig.~\ref{fig:timedurations}}. 
However, often it is impossible or undesirable to predetermine the number of steps needed to stop the iterations. 
Towards this end, we deploy an algorithm that allows the nodes to distributively stop iterations in a finite number of steps, tolerating some deviation from the exact optimal solution. 
Before we proceed with the finite time implementation, we make the following assumption:
\begin{assumption}\label{assumption:3}
    The diameter of the network $D$ is known to all server nodes.
\end{assumption}

{Under Assumption~\ref{assumption:3},} Cady \emph{et al.} in~\cite{cady_finite-time_2015} proposed an algorithm which is based on the ratio-consensus protocol~\cite{dominguez-garcia_coordination_2010} and takes advantage of $\min$- and $\max$-consensus iterations to allow the nodes to determine the time step, $k_0$, when their ratios, $\{\mu_{j}[k_0] | v_{j}\in \mathcal{V} \}$, are within $\epsilon$ of each other.

{First, we present} the synchronous case, {in order to demonstrate the main idea before we present the asynchronous case. Towards this end,} we adopt the algorithm proposed by Cady \emph{et al.} in \cite{cady_finite-time_2015} \emph{mutatis mutandis}. 
More specifically, the algorithm makes use of the following ideas:
\begin{list5}
    \item {
        Each node $v_{j}$ runs ratio consensus iteration, as described in Lemma~\ref{lemma_christoforos}; in our case, we use initial conditions $y_{j}[0]=\ell_j+u_{j}$ and $z_{j}[0]=\pi_j^{\max}$.
    } 
    \item {
        At the same time, each node maintains two auxiliary states, $m_j[k]$ and $M_j[k]$, which are updated using $\min$- and $\max$-consensus, respectively.
    }
    \item {
        Every $D$ steps (where $D$ is the diameter of the graph) each node checks whether $|M_j[k]-m_j[k]|<\epsilon$. 
        If this is the case, then the ratios for all nodes are close to the asymptotic value and it stops iterating. 
        Otherwise, $m_j[k]$ and $M_j[k]$ are reinitialized to {$\mu_{j} [k]$}.
    }
\end{list5}
The algorithm, adopted to our case, is described in Algorithm~\ref{Algorithm_finitetimeRC} for digraphs (which means it holds for undirected graphs as well, that we consider in this case).
\begin{algorithm}
\caption{Distributed Finite-Time Termination for Ratio Consensus}
\begin{algorithmic}
    \STATE \textbf{Input:} A strongly connected digraph $\mathcal{G}=(\mathcal{V}, \mathcal{E})$. Each node $v_{j} \in \mathcal{V}$ knows its out-degree $\mathcal{N}_j^{+}$.
    Initial values are $y_{j}[0]=\ell_j+u_{j}$ and $z_{j}[0]=\pi_j^{\max}$, and tolerance $\epsilon$.
    \STATE \textbf{set} $M_j[0]=+\infty, ~ m_j[0]=-\infty,  {\rm flag}_j[0]=0, \mu_{j}= \frac{y_{j}[0]}{z_{j}[0]}$
    \STATE \textbf{set} $p_{lj}=\frac{1}{1+d_j^{\rm out}}$, $\forall~v_l\in \mathcal{N}_j^{+} \cup \{v_{j}\}$ (zero otherwise)
    \FOR{$k \geq 0$}
        \WHILE{${\rm flag}_j[k]=0$}
            \IF{$k \mod D =0$ and $k\neq 0$}
                \IF{$|M_j[k]-m_j[k]|< \epsilon$}
                    \STATE \textbf{set} ${\rm flag}_j[k] = 1$
                \ENDIF
                \STATE \textbf{set} $M_j[k]=m_j[k]=\mu_{j}[k] = \frac{y_{j}[k]}{z_{j}[k]}$
            \ENDIF
            \STATE  \textbf{broadcast} to all $v_l \in \mathcal{N}_j^{+}$: \newline $p_{lj}y_{j}[k]$, $p_{lj}z_{j}[k]$, $M_j[k]$, $m_j[k]$
            \STATE  \textbf{receive} from all $v_{i} \in \mathcal{N}_j^{-}$: \newline $p_{ji}y_{i}[k]$, $p_{ji}z_i[k]$, $M_i[k]$, $m_i[k]$
            \STATE  \textbf{compute} \\
            \STATE $y_{j}[k]\leftarrow \sum_{v_{i}\in \mathcal{N}_j^{-} \cup \{v_{j}\}}p_{ji}y_{i}[k]$
            \STATE $z_{j}[k]\leftarrow  \sum_{v_{i}\in \mathcal{N}_j^{-} \cup \{v_{j}\}}p_{ji}z_{i}[k]$
            \STATE $M_j[k]\leftarrow \max_{v_{i}\in \mathcal{N}_j^{-} \cup \{v_{j}\}} M_i[k]$
            \STATE $m_j[k]\leftarrow \min_{v_{i}\in \mathcal{N}_j^{-} \cup \{v_{j}\}} m_i[k]$
        \ENDWHILE
    \ENDFOR
\end{algorithmic}
\label{Algorithm_finitetimeRC}
\end{algorithm}

\begin{remark}
The number of iterations needed for the distributed algorithm to terminate at optimization step $m$, $T_c[m]$, is a multiple of the diameter of the network. 
As it will be shown in the simulations, the distributed algorithm converges fast and it only needs a fraction of the optimization step of horizon $T_h$; see an illustration in~\Cref{fig:timedurations}.
\end{remark}

\begin{figure}[t]
    \includegraphics[width=0.97\columnwidth]{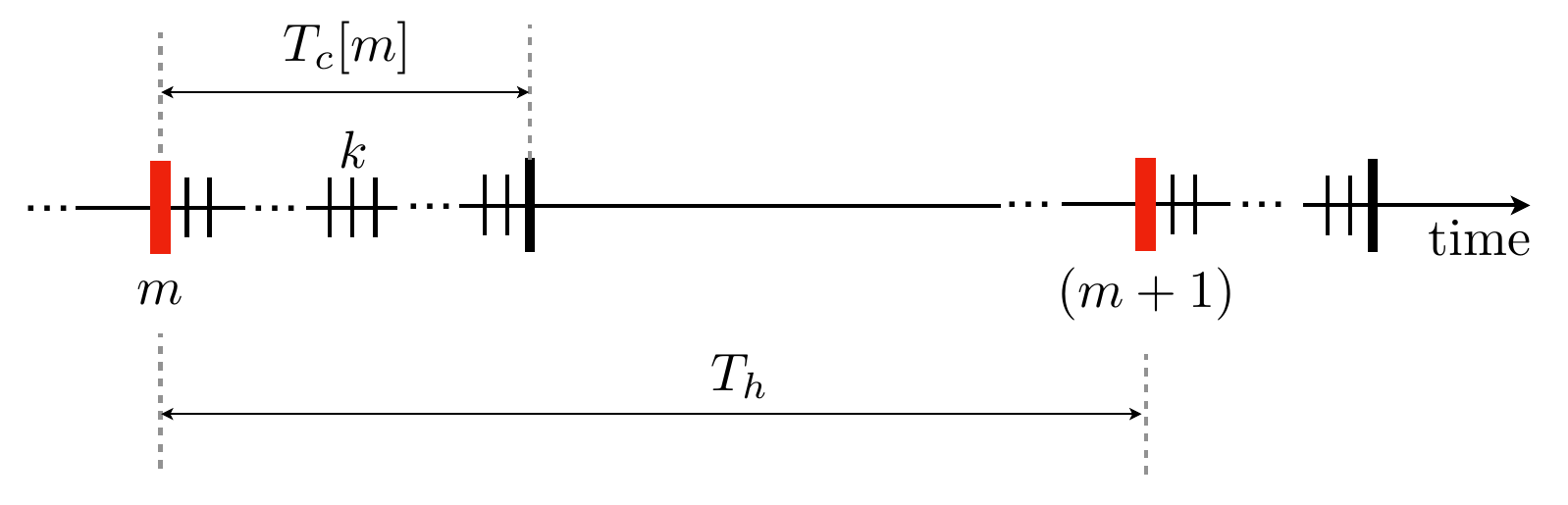}
    \caption{
    At every optimization step of horizon $T_{h}$, the resource allocation optimization requires $T_{c}[m]$ steps to converge. 
    Note that $T_c[m]$ is much smaller in duration that the time horizon of the optimization.
    }
    \label{fig:timedurations}
\end{figure}

\section{An asynchronous distributed algorithm}\label{sec:distributed_asynchronous}

Resource allocation in data centers gives rise to large-scale problems and networks, which naturally call for asynchronous solutions.
Let $t(0)\in \mathbb{R}_{+}$ the time at which the iterations for the optimization start. 
We assume that there is a set of times $\mathcal{T}=\{t(1), t(2),t(3),\ldots\}$ at which one or more nodes transmit some value to their neighbors. 
A message that is received at time $t(k_1)$ and processed at time $t(k_2)$, $k_2>k_1$, experiences a process delay of $t(k_1)-t(k_2)$ (or a time-index delay $k_2-k_1$). 
In Fig.~\ref{fig:async}, we show through a simple example how the time steps evolve for each node in the network; with $t_j(k)$ we denote the time step at which iteration $k$ takes place for node $v_{j}$.
\begin{figure}[ht]
    \includegraphics[width=0.97\columnwidth]{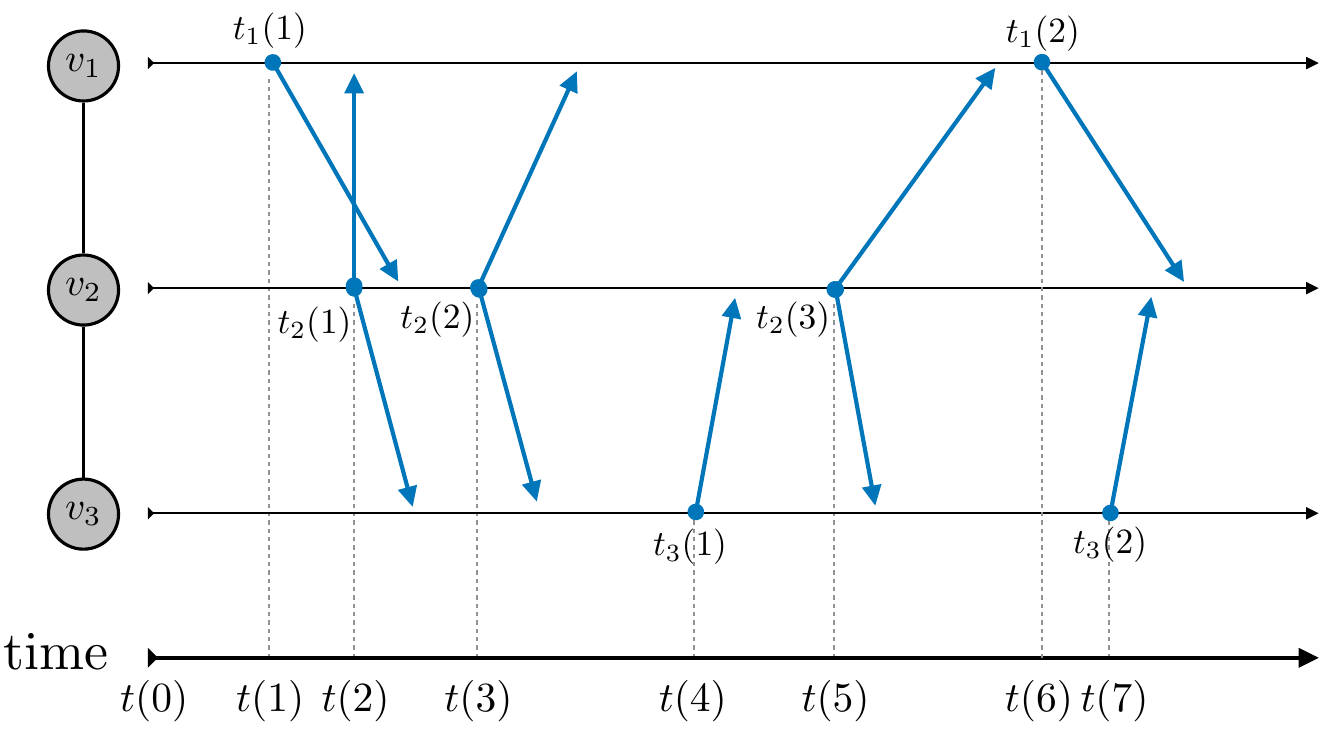}
    \caption{
    A simple example of a network consisting of 3 nodes. 
    In the timeline of each node, blue ticks indicate an iteration for node $v_{i}$ and the arrows indicate the transmissions. 
    The time in between transmissions is the processing delay. The time from the beginning of the transmission to the end (arrow) is the transmission delay.
    }
    \label{fig:async}
\end{figure}

\begin{assumption}
There exists an upper bound $B$ on the time-index steps that is needed for a node to process the information received from another node.
\end{assumption}

\subsection{Asynchronous \texorpdfstring{$\max-$consensus}{max-consensus}}

When the updates are asynchronous, for any node $v_{j}\in \mathcal{V}$, the update rule is as follows~\cite{giannini_convergence_2013}:
\begin{align*}
    x_j[t_j(k+1)] = \max_{v_{i}\in \mathcal{N}_j^{-}[t_j(k+1)] \cup \{v_{j}\}}\{ x_i[t_j(k)+\theta_{ij}(k)] \},
\end{align*}
where $x_i[t_j(k)+\theta_{ij}(k)]$ are the states of the in-neighbors $\mathcal{N}_j^{-}[t_j(k+1)]$ available at the time of the update. 
Variable $\theta_{ij}(k) \in \mathbb{R}$, evaluated with respect to the update time $t_j(k)$, is used here to express asynchronous state updates occurring at the neighbors of node $v_{j}$, between two consecutive updates of the state of node $v_{j}$.
It has been shown in~\cite[Lemma 5.1]{giannini_convergence_2013} that this algorithm converges to the maximum value among all nodes in a finite number of steps $s$, $s \leq B D$.

\subsection{Asynchronous (Robustified) ratio consensus}

An adaptation of the above approach to a protocol where each node updates its information state $x_{j}[k+1]$ by combining the available (possibly delayed) information received by its neighbors $x_{i}[s]$ ($s\in \mathbb{Z}, s\leq k,~v_{i}\in \mathcal{N}^{-}_j$) using constant positive weights $p_{ji}$ was developed in~\cite{hadjicostis_average_2014}. 
Integer $\bar{\tau}_{ji}[k] \geq 0$ is used to represent the delay of a message sent from node $v_{i}$ to node $v_{j}$ at time instant $k$. 
We require that $0\leq \tau_{ji}[k] \leq \bar{\tau}_{ji} \leq \bar{\tau}$ for all $k\geq 0$ for some finite $\bar{\tau} = \max\{ \bar{\tau}_{ji} \}$, $\bar{\tau} \in \mathbb{Z}_{+}$. 
We make the reasonable assumption that $\tau_{jj}[k]=0$, $\forall v_{j} \in \mathcal{V}$, at all time instances $k$ (i.e., the own value of a node is always available without delay). 
Each node updates its information state according to:
\begin{align*}
    x_{j}[k+1] &=p_{jj}x_{j}[k] + \sum_{v_{i} \in \mathcal{N}^{-}_j} \sum_{r=0}^{\bar{\tau} } p_{ji} x_{i}[k-r]I_{k-r,ji}[r] , 
\end{align*}
\noindent for $k\geq 0$, where  $x_{j}[0] \in \mathbb{R}$ is the initial state of node $v_{j}$; $p_{ji}$ $\forall \varepsilon_{ji} \in \mathcal{E}$  form $P=[p_{ji}]$ that adheres to the graph structure, and is primitive column stochastic; and
\begin{align}\label{eq:indicatorfunction}
    I_{k,ji}(\tau) =
    \begin{cases} 1, & \text{if $\tau_{ji}[k] =\tau$,}
    \\
    0, &\text{otherwise.}
    \end{cases}
\end{align}

\begin{lemma}~\cite[\emph{Lemma 2}]{hadjicostis_average_2014}
\label{our_lemma}
Consider a strongly connected digraph $\mathcal{G}(\mathcal{V}, \mathcal{E})$. Let $y_{j}[k]$ and $z_{j}[k]$ (for all $v_{j} \in \mathcal{V}$ and $k=0,1,2,\ldots$) be the result of the iterations
\begin{align}\label{eq:ratio-delays}
    y_{j}[k+1] &=p_{jj}y_{j}[k] + \sum_{v_{i} \in \mathcal{N}^{-}_j} \sum_{r=0}^{\bar{\tau} } y_{ji}[k-r]I_{k-r,ji}[r] \; , \\
    z_{j}[k+1] &=p_{jj}z_{j}[k] + \sum_{v_{i} \in \mathcal{N}^{-}_j} \sum_{r=0}^{\bar{\tau} } z_{ji}[k-r]I_{k-r,ji}[r] \; ,
\end{align}
with $y[0]=(y_0(1) \ \ y_0(2) \ \ldots \ y_0(|\mathcal{V}|))^T\equiv y_0$ and  $z[0]=\mathbb{1}$; $I_{k,ji}$ is an indicator function that captures the bounded delay $\tau_{ji}[k]$ on link $(v_{j}, v_{i})$ at iteration $k$ (as defined in~\eqref{eq:indicatorfunction}, $\tau_{ji}[k] \leq \bar{\tau}$). 
Then, the solution to the average consensus problem can be asymptotically obtained as
$$
\displaystyle \lim_{k\rightarrow \infty} \mu_{j}[k]=\frac{\sum_{v_\ell \in \mathcal{V}} y_0(\ell)}{|\mathcal{V}|} \; , \; \forall v_{j} \in \mathcal{V} \; ,
$$
where $\displaystyle \mu_{j}[k]={y_{j}[k]}/{z_{j}[k]}$.
\end{lemma}

\subsection{Finite-time asynchronous ratio consensus}

As it is the case for the synchronous distributed algorithm (see \S~\ref{sec:distributed_synchronous}), the consensus algorithm should terminate before the next optimization step and in a distributed fashion. 
In what follows, we propose a distributed termination protocol for the asynchronous case, based on the one used for the synchronous case. 
{We believe, that this is the first termination algorithm that can handle delays and perform asynchronous consensus.}

The proposed termination algorithm has the same principles as before~\cite{cady_finite-time_2015}. 
{However, in order to make the ideas put forth in~\cite{cady_finite-time_2015} applicable into the asynchronous case we expand upon them using several innovations. More concretely, when compared to the synchronous case the aforementioned innovations are outlined below}:
\begin{list5}
    \item {
        The $\min$ and $\max-$consensus algorithm converge in $(1+\bar{\tau})D$ steps~\cite{giannini_convergence_2013}.
    }
    \item {
        Every $(1+\bar{\tau})D$ steps each  node  checks  whether $|M_j[k]-m_j[k]|< \epsilon$. 
        If  this is the case, then the ratios for all nodes are close to the asymptotic value and it stops iterating. Otherwise, $m_j[k]$ and $M_j[k]$ are reinitialized to $\mu_{j}[k]$.
    }
\end{list5}

The algorithm is described in Algorithm~\ref{Algorithm_finitetimeRCwDelays} or digraphs; note that this implies it also holds for \textit{undirected} graphs as well, that we consider in this case.
\begin{algorithm}
\caption{Distributed Finite-Time Termination for Asynchronous Ratio Consensus}
\begin{algorithmic}
\STATE \textbf{Input:} A strongly connected digraph $\mathcal{G}=(\mathcal{V}, \mathcal{E})$. Each node $v_{j} \in \mathcal{V}$ knows its out-degree $\mathcal{N}_j^{+}$.
Initial values are $y_{j}[0]=\ell_j+u_{j}$ and $z_{j}[0]=\pi_j^{\max}$, and tolerance $\epsilon$.
\STATE \textbf{set} $M_j[0]=+\infty, ~ m_j[0]=-\infty,  {\rm flag}_j[0]=0, \mu_{j}= \frac{y_{j}[0]}{z_{j}[0]}$
\STATE \textbf{set} $p_{lj}=\frac{1}{1+d_j^{\rm out}}$, $\forall~v_l\in \mathcal{N}_j^{+} \cup \{v_{j}\}$ (zero otherwise)
\FOR{$k \geq 0$}
\WHILE{${\rm flag}_j[k]=0$}
\IF{$k \mod (1+\bar{\tau})D =0$ and $k\neq 0$}
\IF{$|M_j[k]-m_j[k]|< \epsilon$}
\STATE \textbf{set} ${\rm flag}_j[k] = 1$
\ENDIF
\STATE \textbf{set} $M_j[k]=m_j[k]=\mu_{j}[k] = \frac{y_{j}[k]}{z_{j}[k]}$
\ENDIF
\STATE  \textbf{broadcast} to all $v_l \in \mathcal{N}_j^{+}$: \newline $p_{lj}y_{j}[k]$, $p_{lj}z_{j}[k]$, $M_j[k]$, $m_j[k]$
\STATE  \textbf{receive} from all $v_{i} \in \mathcal{N}_j^{-}[k]$: \newline $p_{ji}y_{i}[k]$, $p_{ji}z_i[k]$, $M_i[k]$, $m_i[k]$
\STATE  \textbf{compute} \\
\STATE $y_{j}[k]\hspace{-0.1cm}\leftarrow \hspace{-0.1cm} p_{jj}y_{j}[k] + \sum_{v_{i} \in \mathcal{N}^{-}_j} \sum_{r=0}^{\bar{\tau} } y_{ji}[k-r]I_{k-r,ji}[r]$
\STATE $z_{j}[k]\hspace{-0.1cm}\leftarrow \hspace{-0.1cm} p_{jj}z_{j}[k] + \sum_{v_{i} \in \mathcal{N}^{-}_j} \sum_{r=0}^{\bar{\tau} } z_{ji}[k-r]I_{k-r,ji}[r]$
\STATE $M_j[k]\hspace{-0.1cm}\leftarrow \hspace{-0.1cm} \max_{v_{i}\in \mathcal{N}_j^{-} \cup \{v_{j}\}}\{ M_i[t_j(k)+\theta_{ij}(k)] \}$
\STATE $m_j[k]\hspace{-0.1cm}\leftarrow \hspace{-0.1cm} \max_{v_{i}\in \mathcal{N}_j^{-} \cup \{v_{j}\}}\{ m_i[t_j(k)+\theta_{ij}(k)] \}$
\ENDWHILE
\ENDFOR
\end{algorithmic}
\label{Algorithm_finitetimeRCwDelays}
\end{algorithm}

\begin{theorem}\label{thm:finitercdelays}
Algorithm~\ref{Algorithm_finitetimeRCwDelays} converges in finite time.
\end{theorem}

\begin{proof}
    From Lemma~\ref{our_lemma}, we know that $\lim_{k\rightarrow \infty} \mu_{j}[k]=(\sum_{v_\ell \in \mathcal{V}} y_0(\ell))/\mathcal{V}$, for all $v_{j} \in \mathcal{V}$. 
    Therefore, it follows that
    \begin{align}
        \lim_{k\rightarrow \infty}\left|\max_{v_{j} \in \mathcal{V}}\mu_{j}[k]-\frac{\sum_{v_\ell \in \mathcal{V}} y_0(\ell)}{|\mathcal{V}|} \right|=0,  
    \end{align}
    which means that essentially $\lim_{k\rightarrow \infty} M[k]=\frac{\sum_{v_\ell \in \mathcal{V}} y_0(\ell)}{|\mathcal{V}|}$.
    Additionally, $k_0$ exists, such that for all $k \geq k_0$, we have
    \begin{align}
        \left|\mu_{j}[k]-\frac{\sum_{v_\ell \in \mathcal{V}} y_0(\ell)}{|\mathcal{V}|} \right|<\epsilon, ~\forall~ v_{j} \in \mathcal{V}.
    \end{align}
    Therefore, it follows that 
    \begin{align}
        \left|\max_{v_{j} \in \mathcal{V}}\mu_{j}[k]-\frac{\sum_{v_\ell \in \mathcal{V}} y_0(\ell)}{|\mathcal{V}|} \right|<\epsilon,
    \end{align}
    In turn, this implies that there exists $k_0$, such that for all $k\geq k_0$, 
    \begin{align}
        \left|M[k]-\frac{\sum_{v_\ell \in \mathcal{V}} y_0(\ell)}{|\mathcal{V}|} \right|<\epsilon.
    \end{align}
    Similar arguments hold for $m[k]$. 
    Since $\{M[r(1+\bar{\tau})D]\}_{r\in\mathbb{N}}$ and $\{m[r(1+\bar{\tau})D]\}_{r\in\mathbb{N}}$ are sub-sequences of sequences that converge (due to the fact that asynchronous $\max-$ consensus converges within $(1+\bar{\tau})D$ steps), then they converge to the same limit. 
    Therefore, there exists $r_0$, such that for all $r\geq r_0$, $|M[r(1+\bar{\tau})D]-m[r(1+\bar{\tau})D]| < \epsilon$.
\end{proof}

\begin{remark}\label{remark:2}
We stress that similar results were proposed in \cite{prakash_distributed_2020} for guaranteeing convergence to approximate average consensus in a finite number of steps, allowing for time-varying bounded delays in information transmission and reception between agents. Nevertheless, apart from the fact that our results are obtained for an optimization problem for CPU scheduling, there are some additional differences:
\begin{list5}
\item we use the consensus algorithm in the concept of asynchronous
operation, rather than synchronous operation with delays, despite the
fact that the mathematical analysis relies on similar concepts;
\item the window used for updating the min/max value of the agents is
different (for us this is $(1+\bar{\tau})D$ while for them is
$(1+\bar{\tau})D +\bar{\tau}$), and
\item we show via simulation that the lemmas (and, hence, the proofs) in \cite{prakash_distributed_2020}
are incorrect (see also the discussion in Section~\ref{sec:discussion}).
\end{list5}
\end{remark}

\section{Simulations}\label{sec:simulations}

To validate our scheme, we divide our evaluation into three separate segments. 
The first focuses on simulating the performance using a simple, easy to understand, network of five nodes. 
The second one presents a thorough quantitative evaluation using simulations for various randomly generated graphs and latencies. 
The last one, provides a large scale evaluation with network graphs and simulation parameters that would be applicable in large scale data centers having thousands of nodes.
To our knowledge this is the first work that tackles the problem at this scale in this setting while also providing a thorough evaluation and theoretical guarantees.
All experiments are computed on a workstation using an AMD 3970X CPU with $32$ cores at $4.0$GHz, $256$ GB $3600$ MHz DDR4 RAM, and Matlab R2022b (build 9.13.0.2080170)\footnote{To foster reproducibility both code and datasets used for our numerical evaluation are publicly available at: \url{https://github.com/andylamp/federated-capacity-consensus}.}.

\subsection{Evaluation using a small network}

The digraph is comprised out of $|\mathcal{V}|=5$ vertices and has a diameter equal to $D=4$; for helping exposition the exact digraph is shown in Fig.~\ref{fig:example_digraph}. 

\begin{figure}[ht]
    \centering
    \includegraphics[width=0.45\columnwidth]{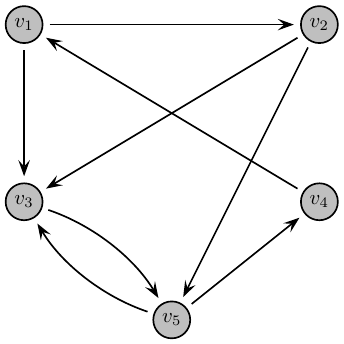}
    \caption{The strongly connected digraph network comprised out of five nodes which is used to evaluate the validity of our results though an indicative, small-scale example.}
    \label{fig:example_digraph}
\end{figure}
All node are set with equal capacities and the workload vector $\rho$ is set to $\rho=[1, 2, 3, 4, 5]$ in all runs.
Further, we set the convergence threshold for the absolute difference of the quantity $|M_j[k]-m_j[k]| < \epsilon$ to $\epsilon=10^{-5}$.
Throughout our experiments, at each interval the workload to be scheduled is generated for each node independently. 
Concretely, each randomly chooses a job cost from a uniform distribution bounded between an acceptable cost range which is provided upon initialisation.
These values are then concatenated to generate a workload vector which has a value between that range for each of the nodes.

Then in order to study the impact of increased delay in the number of total iterations required, we evaluate our proposed algorithm when using $\bar{\tau}=[4, 9]$.
We start by showing the results for $\bar{\tau}=4$ in~Fig.~\ref{fig:async-basic-converge-delay-5}. 
In this figure, we observe that converge happens after $120$ iterations which is $4{(1+\bar{\tau})}D$, meaning that in total four rounds are required. 

\begin{figure}[ht]
    \centering
    \includegraphics[width=\columnwidth]{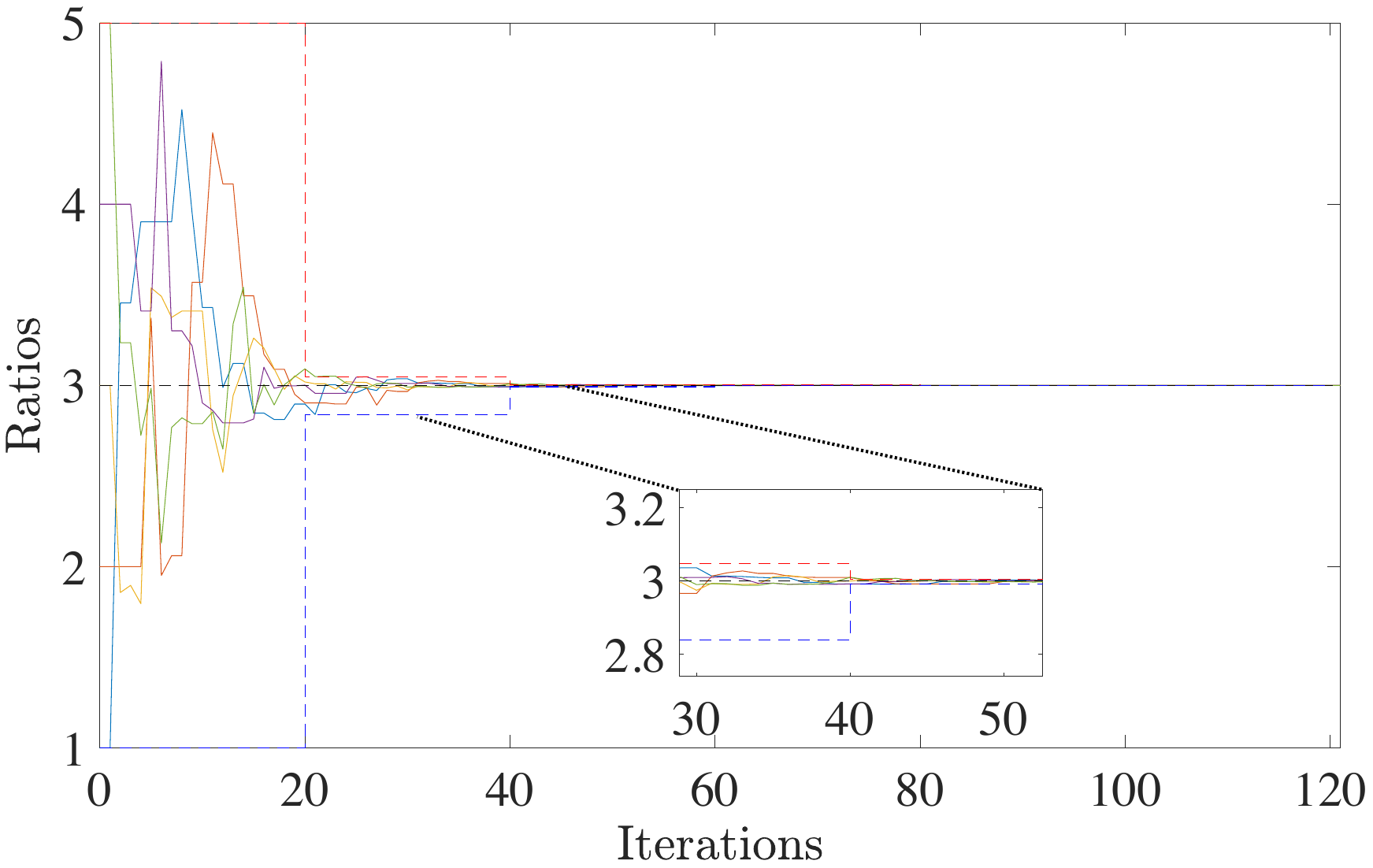}
    \caption{
    A simple example of a network of five nodes as described in~Fig.~\ref{fig:example_digraph} when the links experience time-varying delays with maximum delay ($\bar{\tau}$) of $4$. 
    The figure shows the evolution of the converge ratios across all nodes along with the $\min-$consensus (dashed blue) and the $\max-$consensus (dashed red).}
    \label{fig:async-basic-converge-delay-5}
\end{figure}

Following, we shift our attention to~\Cref{fig:async-basic-converge-delay-10} in which we show the results of the same experiment when using a delay value of $\bar{\tau}=9$.
Concretely, we see that the increased delay has an impact on the total iterations required to converge increasing them by a factor of about $\approx 1.6$ when compared to the previous experiment.

\begin{figure}[ht]
    \centering
    \includegraphics[width=\columnwidth]{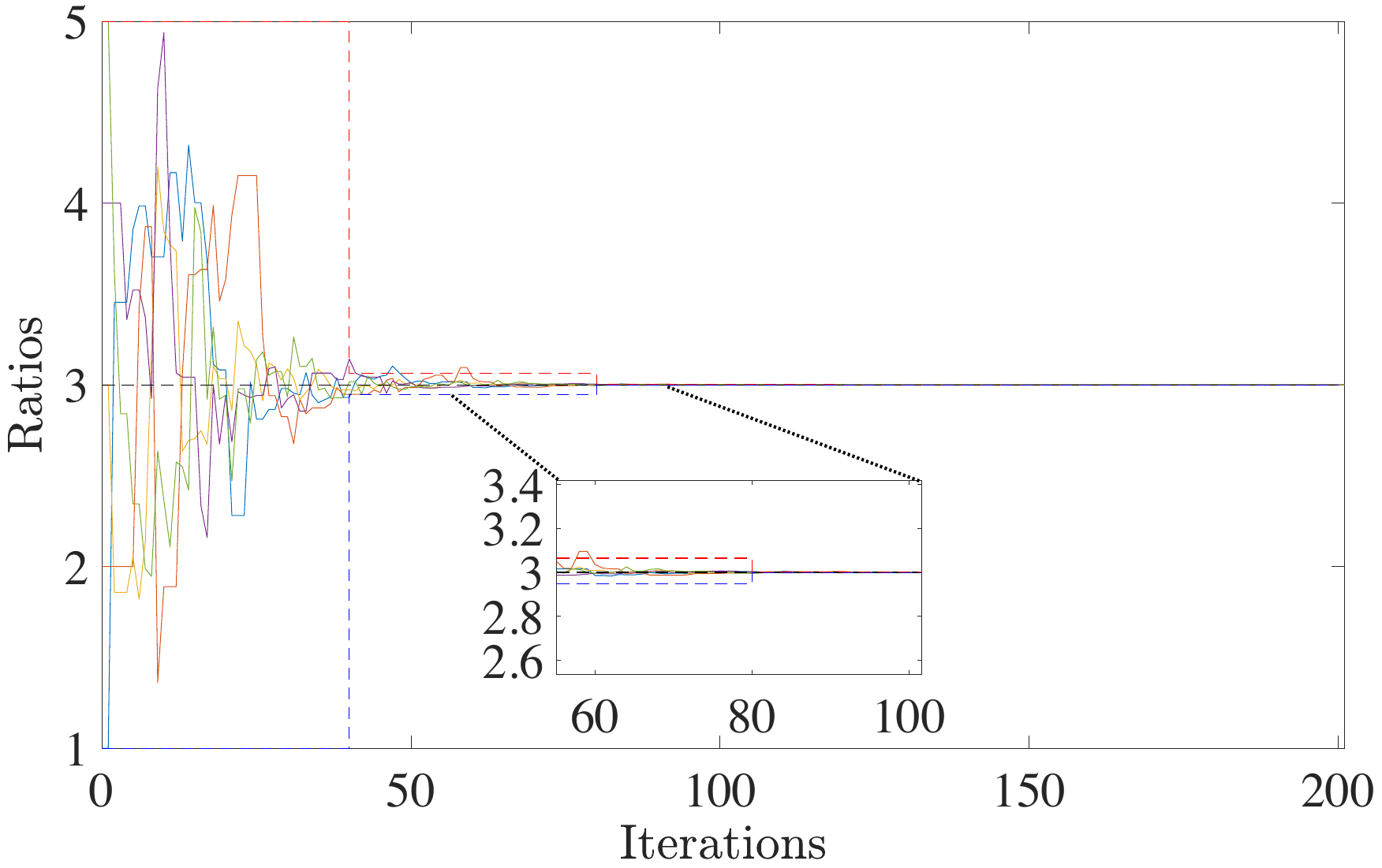}
    \caption{
    Converge ratios, when using the network of five nodes as described in Fig.~\ref{fig:example_digraph} when the links experience time-varying delays with maximum delay ($\bar{\tau}$) of $9$. 
    The $\min-$consensus and $\max-$consensus are depicted by the dashed blue and red lines respectively. 
    }
    \label{fig:async-basic-converge-delay-10}
\end{figure}

We see that both figures converge in multiples of $(1+\bar{\tau})D$ which requires six rounds when having $\bar{\tau}=5$ and four rounds when using $\bar{\tau}=10$. 
Notably, as delay grows the round size increases linearly assuming we operate on the same graph (hence the diameter $D$ remains the same). 
Indeed, the round size for $\bar{\tau}=5$ is $20$ iterations whereas in the case of $\bar{\tau}=10$ the round size is $40$ iterations.
Quantitatively speaking, we observe that as the round size increases the number of rounds required to converge decreases.
We conjecture that this can be attributed to the fact that as the round size increases the information has an elongated iteration window to propagate throughout the graph which in turns helps to converge with fewer rounds.
However, since the results are simulated centrally even if the aggregated simulation cost is large, the amortised cost (e.g. the actual computation that would be required per node) is practically very low - even in the presence of large delays.

\begin{remark}
Note that there are some nodes $v_j\in\mathcal{V}$ for which the state $\mu_j[k']$ is larger than the maximum $M(k)$, where $k'>k$ and $k\mod D =0$ (note that this constitutes a counterexample to Lemma IV.2 in \cite{khatana_gradient-consensus_2020}). 
Despite the fact that the ratio is not monotonically decreasing (due to the nonlinearity imposed by the ratio), the main properties that guarantee the convergence of this algorithm is that the ratio is guaranteed to converge and the $\max$-consensus algorithm converges within ${(1+\bar{\tau})}D$ steps.  
\end{remark}

\subsection{Evaluation using varying delays and network sizes}

The previous example is indicative on how our scheme performs in a tangible, small-scale scenario. 
{In this section, we evaluate the performance of our proposed algorithm across a broader range of parameters reflecting realistic deployments, as such our generated topologies attempt to replicate ones that would be in real data-centers.}
To that end, we create a test suite monitoring both convergence and actual simulation execution time for varying graph sizes and delays.
check this again: Concretely, for a given amount of trials, graph size dictated by $|\mathcal{V}|$, and a range of  delays upper bounds we create a random graph for different unique pairs $\langle |\mathcal{V}|, \; \bar{\tau} \rangle$.
The values considered for graph sizes and delays upper bounds are $|\mathcal{V}|=[20, 50, 100, 200, 300, 600]$ and $\bar{\tau}=[1, 5, 10, 15, 20, 30]$, respectively, which result in the evaluation of $36$ unique $\langle |\mathcal{V}|, \; \bar{\tau} \rangle$ pairs.
More specifically, for each unique $\langle |\mathcal{V}|, \; \bar{\tau} \rangle$ pair we perform $10$ trials and average the results for each pair. 
We also note, that throughout our experiments, . More concretely, as long as we are able to generate a connected random graph, all trial instances converge within the maximum iteration limit set; this value is set to $4000$ iterations across all runs.
{Additionally, while the randomly generated network topologies are guaranteed to be strongly connected and of relatively low diameter they are not necessarily assured to be fat-tree and/or spine-leaf compliant topologies.}
We begin by presenting the number of iterations required to converge, on average, across $10$ runs for each $\langle |\mathcal{V}|, \; \bar{\tau} \rangle$ pair;  results are shown in~Fig.~\ref{fig:converge-iterations}.

\begin{figure}[ht]
    \centering
    \includegraphics[width=.98\columnwidth]{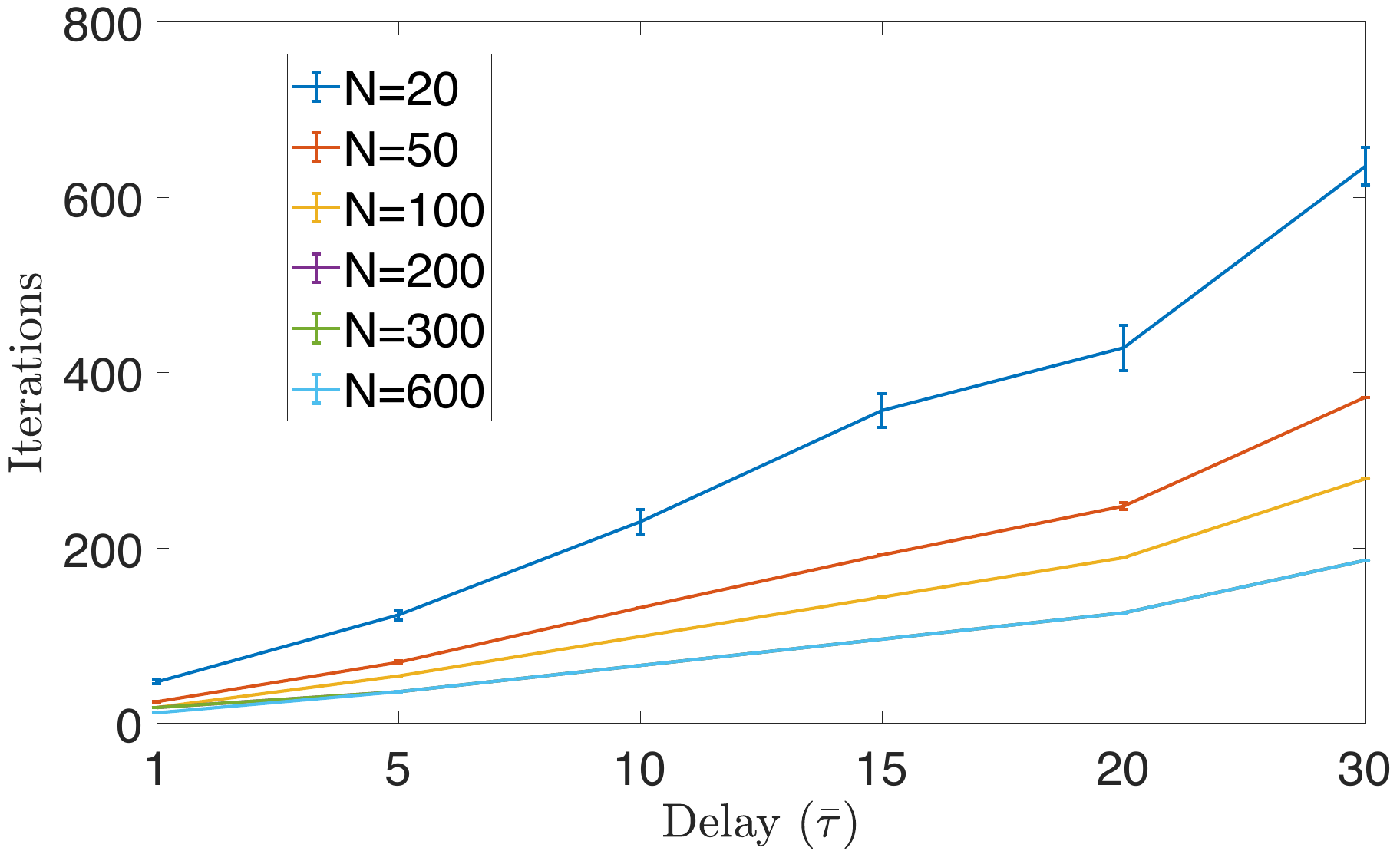}
    \caption{
    Total number of iterations required to converge for each unique $\langle |\mathcal{V}|, \; \bar{\tau} \rangle$ pair averaged across $10$ trials.
    The $x$-axis shows the different delays ($\bar{\tau}$) while each line represents the number of nodes ($|\mathcal{V}|$) that exist within each graph. 
    }
    \label{fig:converge-iterations}
\end{figure}

Fig.~\ref{fig:converge-iterations} indicates that smaller networks require more iterations than larger ones to converge, which are still multiples of $(1+\bar{\tau})D$.
At first glance this observation might seem as counter-intuitive, however, we conjecture that such behaviour is encountered because the round size for smaller networks is smaller thus the system has fewer iterations to reach a steady state within each round.
Indeed, similarly to the delay, recall that each round length is dictated by $(1+\bar{\tau})D$; thus, fixing the delay $\bar{\tau}$ and increasing the diameter $D$---as is the case when the graph network grows---results in linear inflation of the round size. 
Notably, even if the round size increases this does not mean that the execution time is less.
In fact it is quite the opposite since the total simulation time is higher as the network size increases. 
However, the extrapolated actual cost per node is much less. This is because, the workload for each can be parallelized and is asynchronous.

\begin{figure}[ht]
    \centering
    \begin{subfigure}{.99\linewidth}
        \centering
        \includegraphics[width=.98\columnwidth]{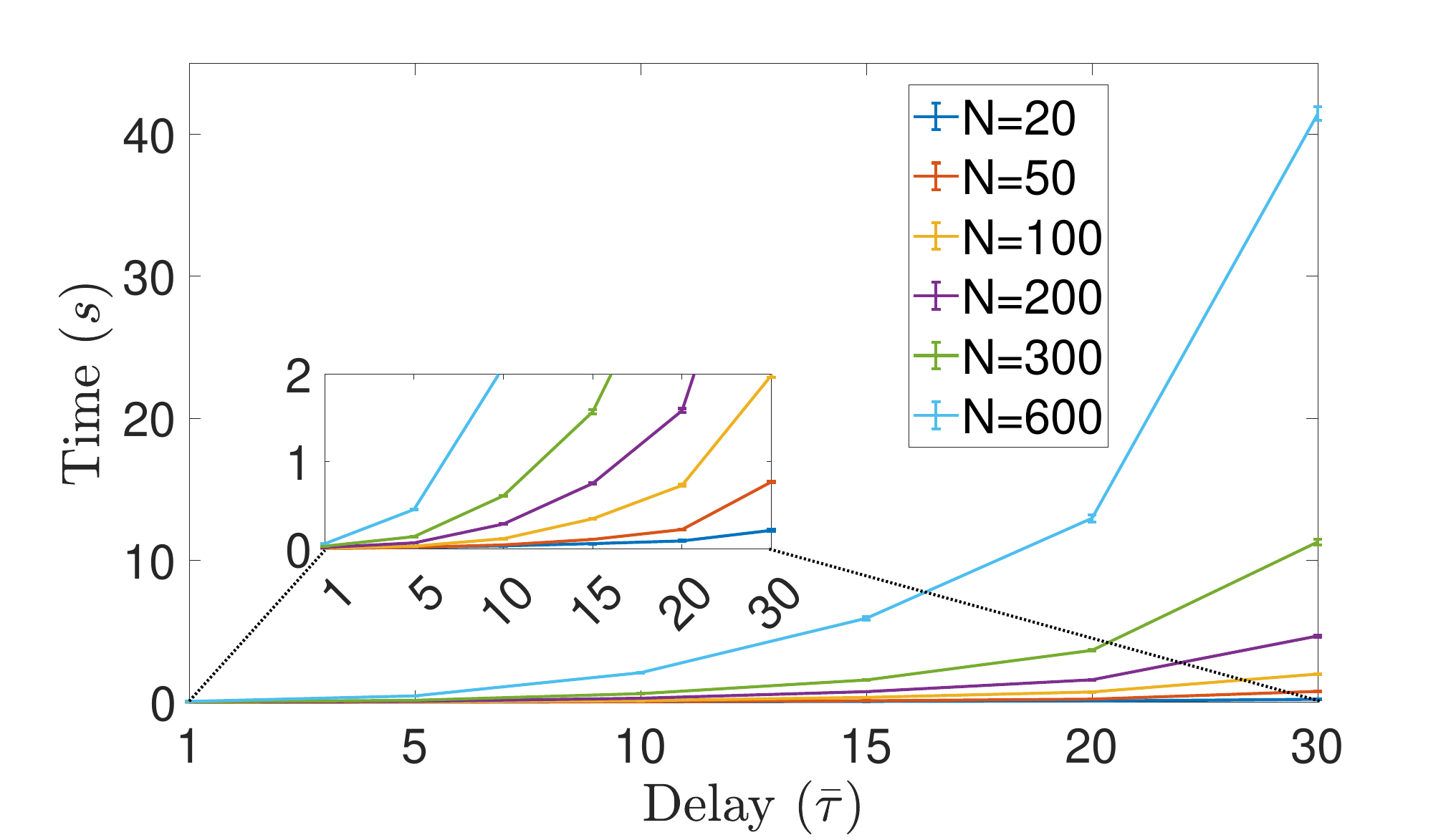}  
        \caption{Total simulation time.}
        \label{fig:converge_sim_time}
    \end{subfigure}
    \begin{subfigure}{.99\linewidth}
        \centering
        \includegraphics[width=.98\columnwidth]{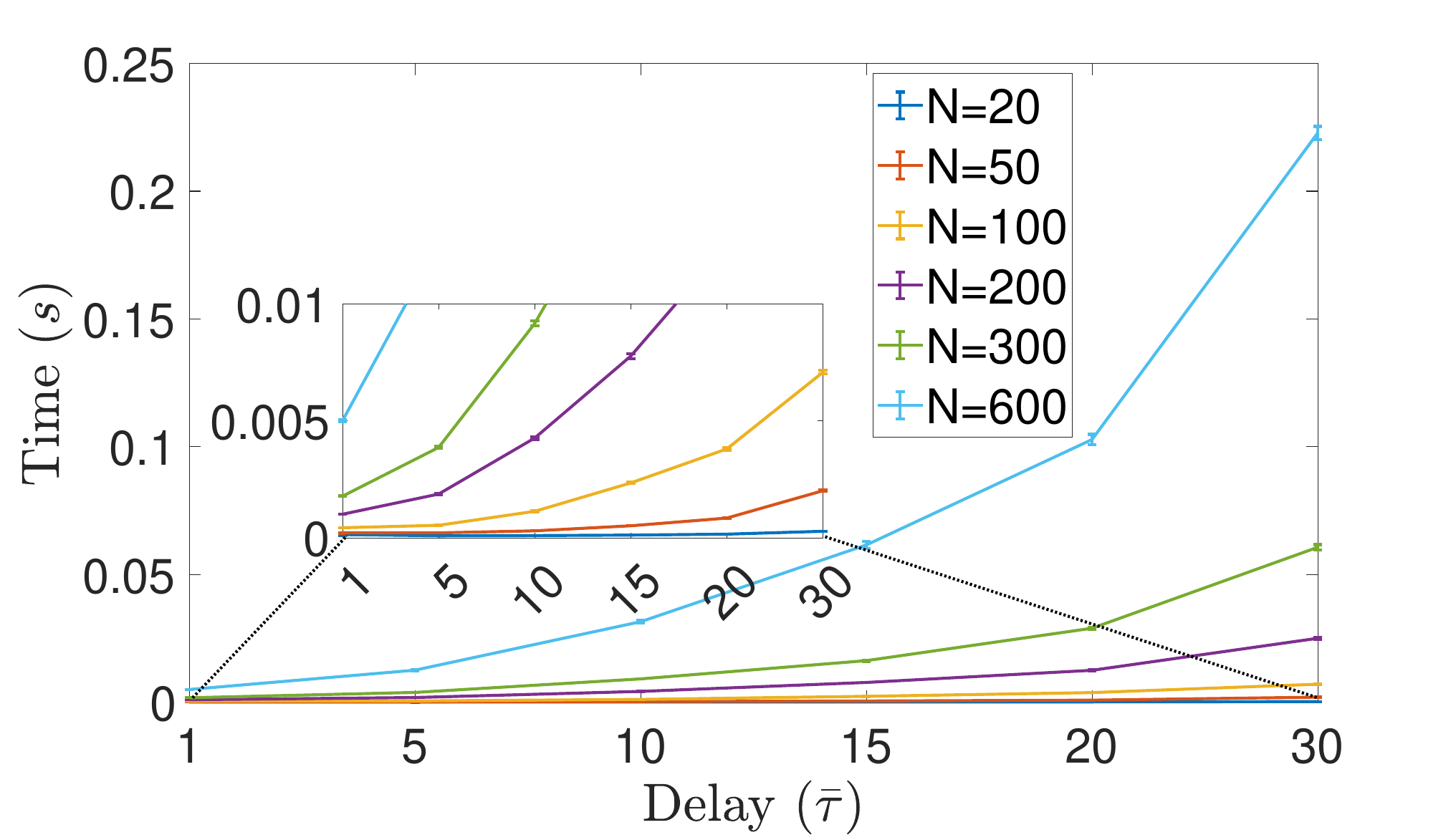}  
        \caption{Iteration time.}
        \label{fig:avg_iter_time}
    \end{subfigure}
    \caption{Top (Fig.~\ref{fig:converge_sim_time}), we present the total execution time required to converge for each unique $\langle |\mathcal{V}|, \; \bar{\tau} \rangle$ pair averaged across $10$ trials.
    The $x$-axis shows the different delays upper bounds ($\bar{\tau}$) while each line represents the number of nodes ($|\mathcal{V}|$) that exist within each graph.
    {Bottom} (Fig.~\ref{fig:avg_iter_time}), we show the average iteration time for each of the different configurations.
    }
    \label{fig:total-execution-time}
\end{figure}

Fig.~\ref{fig:total-execution-time} shows the average execution time required to converge and the average iteration execution time for the same experiments discussed previously.
As we can see from Fig.~\ref{fig:total-execution-time}, the execution time scales exponentially as both delay and graph size increase.
More importantly, this graph shows in practice that larger graphs take more time to converge than smaller ones given the same delay even if the actual rounds to converge are less as graph size increases.
This is because, as we noted previously, even if the iterations are fewer each iteration within a larger graph takes significantly more time to complete in practice.
However, as a general trend we observe that regardless of the network size used in our experiments, if the delay remains below $\bar{\tau}=10$, then it converges relatively quickly. Conversely, it seems that for delays greater than $\bar{\tau}=15$ then the time to converge scales exponentially.

\subsection{Data center scale evaluation}

Previous examples evaluate the performance of the algorithm in practical small-scale deployment. However, these experiments do not capture the scale of modern data centers which contain thousands of server machines.
To that end, to evaluate the data center scalability of our scheme we perform experiments on thousands of nodes.
We assume that in data centers most nodes are few hops away from each other, so we use graphs with a small diameter~\cite{singla_high_2014}.
Further, we assume that the latency within data centers is near zero as shown before in order to satisfy the needs of modern workloads~\cite{guo_pingmesh_2015,alizadeh_less_2012}.
To sum up, in order to provide a realistic data center scale representation, we create a simulation configuration that scales to thousands of nodes; considers graphs of a small diameter; and finally assumes low, even if variable, network delays upper bounds. 
Concretely, the values considered for the graph sizes and delays upper bounds are $|\mathcal{V}|=[20, 200, 500, 1000, 5000, 10000]$ and $\bar{\tau}=[1, 2, 3, 4 5]$ respectively; which result in the evaluation of $30$ unique $\langle |\mathcal{V}|, \; \bar{\tau} \rangle$ pairs.
We note, however, that in order for modern data centers to maintain very low network communication delays, it is desirable to have just a couple of hops between nodes and, hence, we consider graphs with small diameter~\cite{popa_cost_2010, singla_high_2014}.
As previously, for each unique $\langle |\mathcal{V}|, \; \bar{\tau} \rangle$ pair we perform $5$ trials and average the results for each pair. 

\begin{figure}
    \centering
    \includegraphics[width=\columnwidth]{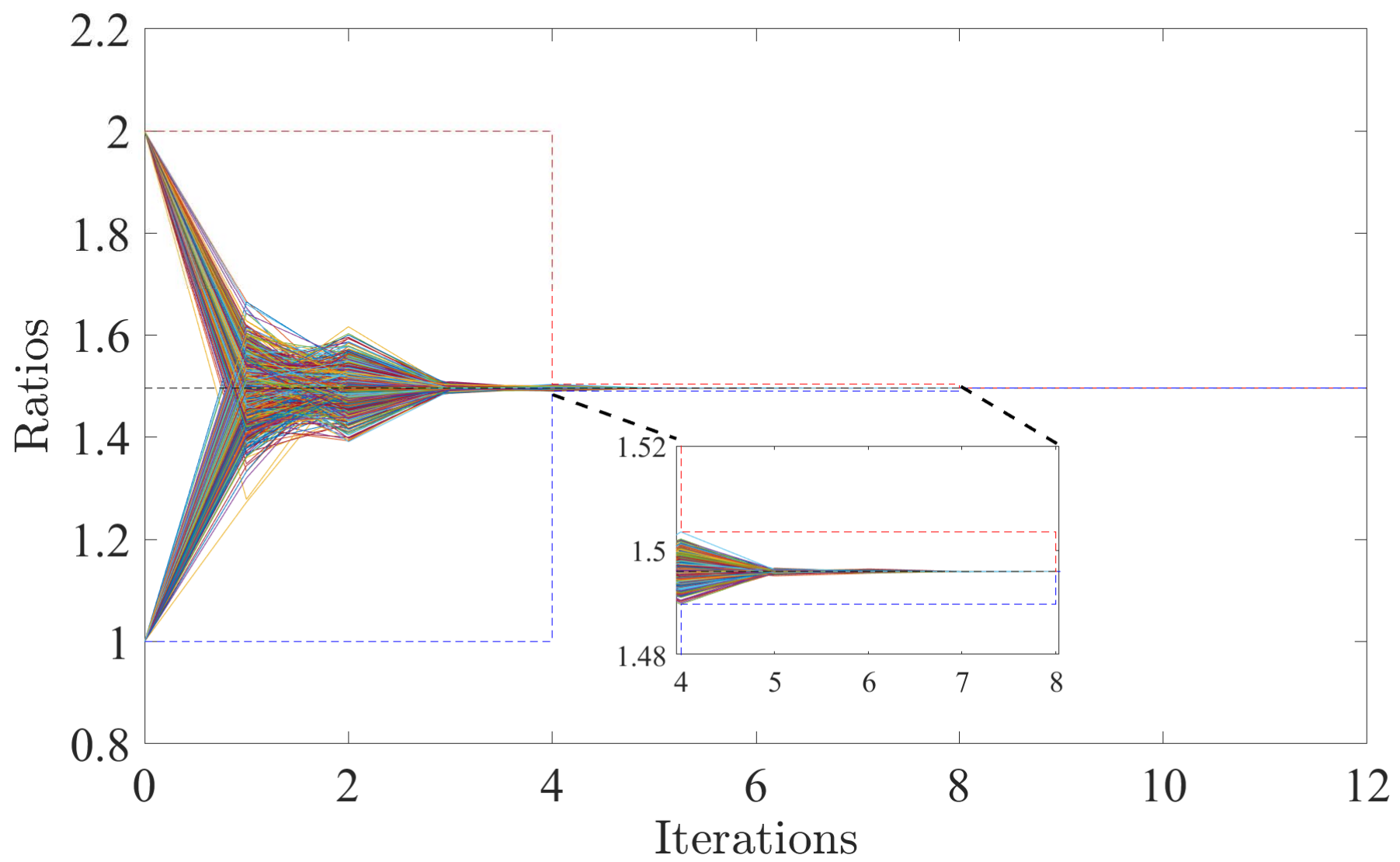}
    \caption{Example run of a network comprised of $1000$ nodes having a diameter equal to $2$ and using a delay upper bound $\bar{\tau}$ of $1$. The network converges to the optimal solution in very few iterations.}
    \label{fig:1k_nodes_run_delay_1}
\end{figure}

Fig.~\ref{fig:1k_nodes_run_delay_1} illustrates the results of an example run of a network size of $1000$ and a delay $\bar{\tau}=1$.
We can see that our scheme is able to converge to the optimal solution in very few iterations. This is attributed to the diameter of the graph which was equal to $D=2$ and to low delays ($\bar{\tau}=1$).

In the next data center scale experiment we vary the number of nodes from $20$ to data center scale of $10000$. We also vary the upper bound on the delay $\bar{\tau}$. Results are shown in~Fig.~\ref{fig:10k_delay_scaling} and~Fig.~\ref{fig:10k_execution_time}. 
Fig.~\ref{fig:10k_delay_scaling} shows the converge scaling with respect to the iterations required as the delays upper bound and network size grow.
Fig.~\ref{fig:10k_execution_time} shows the average total simulation time and per iteration time required per each network size and delays upper bound.
Note, that the simulation indicates the \emph{aggregated times} required to complete each round since for the context of this work we simulate our scheme centrally for all networks.
In practice, in a real system, the actual execution cost per node would be much less since the workload would be executed asynchronously and concurrently.

\begin{figure}
    \centering
    \includegraphics[width=.98\columnwidth]{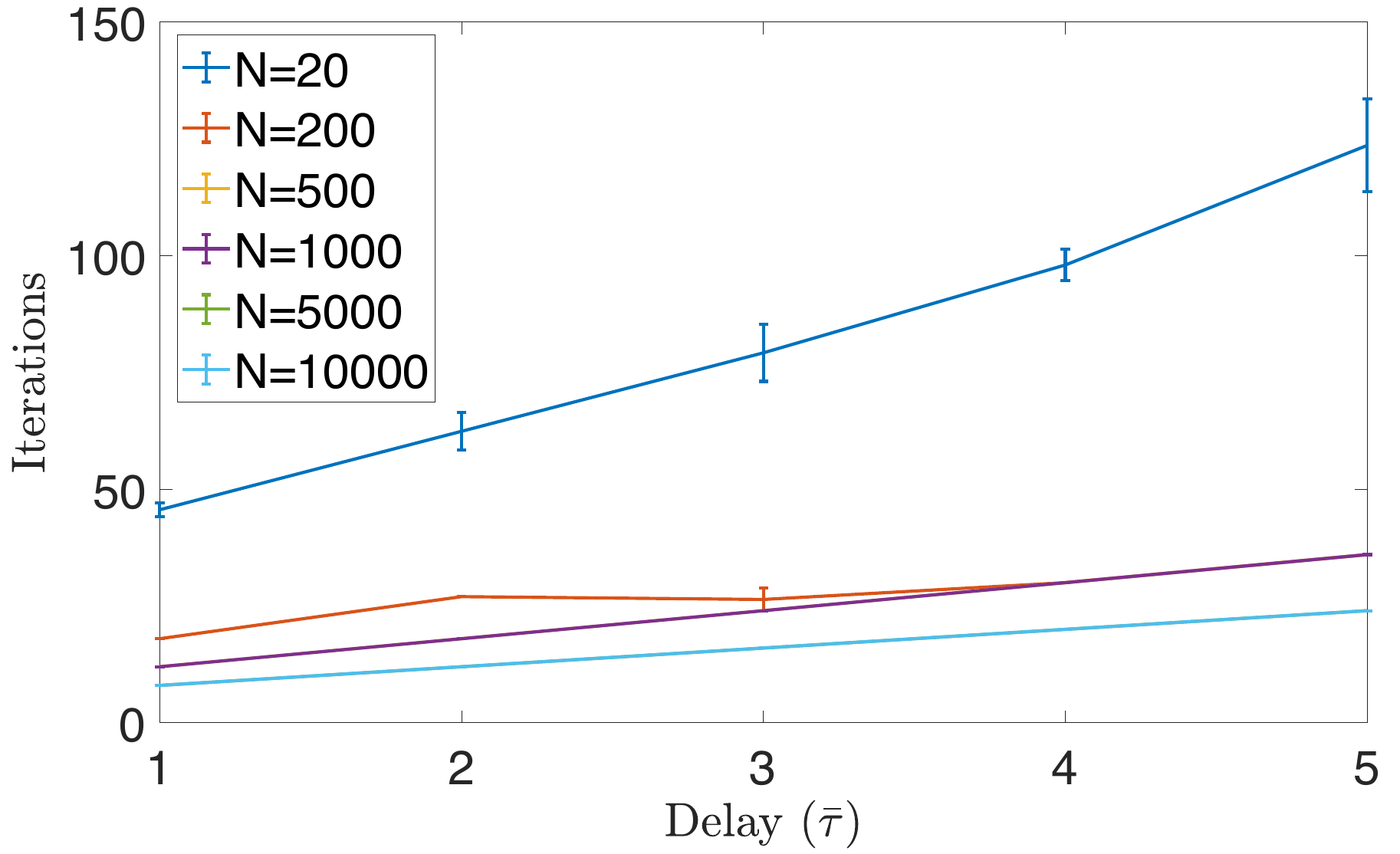}
    \caption{Mean iterations to converge for different network sizes and delay values. Delay plays a larger role in smaller networks (< 200 nodes) whereas as network size increases the delay impact is lower.}
    \label{fig:10k_delay_scaling}
\end{figure}

The same trend can be seen in the converge statistics in Fig.~\ref{fig:10k_converge_stats_delay_1} and Fig.~\ref{fig:10k_converge_stats_delay_5}. We define as the ``min'' the iteration in which the first node successfully converges and the ``max'' the iteration where the last node converges.
Note, that mean is the ``average'' converge iteration for all nodes and the converge ``window'' is the difference between the ``max'' and ``min''.
\begin{figure}
    \centering
    \begin{subfigure}{.49\linewidth}
        \centering
        \includegraphics[scale=0.29]{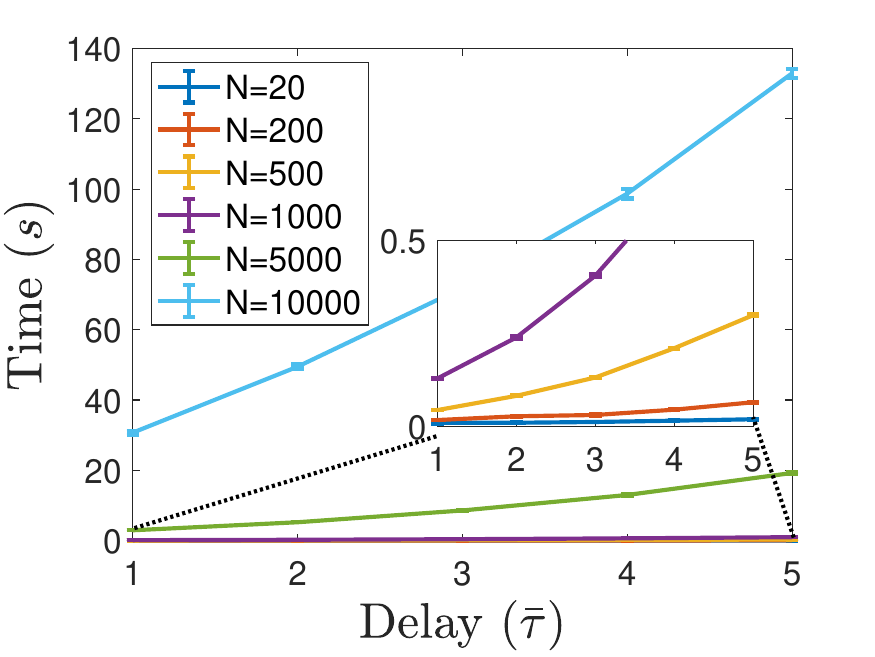}  
        \caption{Total simulation time.}
        \label{fig:10k_converge_sim_time}
    \end{subfigure}
    \begin{subfigure}{.49\linewidth}
        \centering
        \includegraphics[scale=0.29]{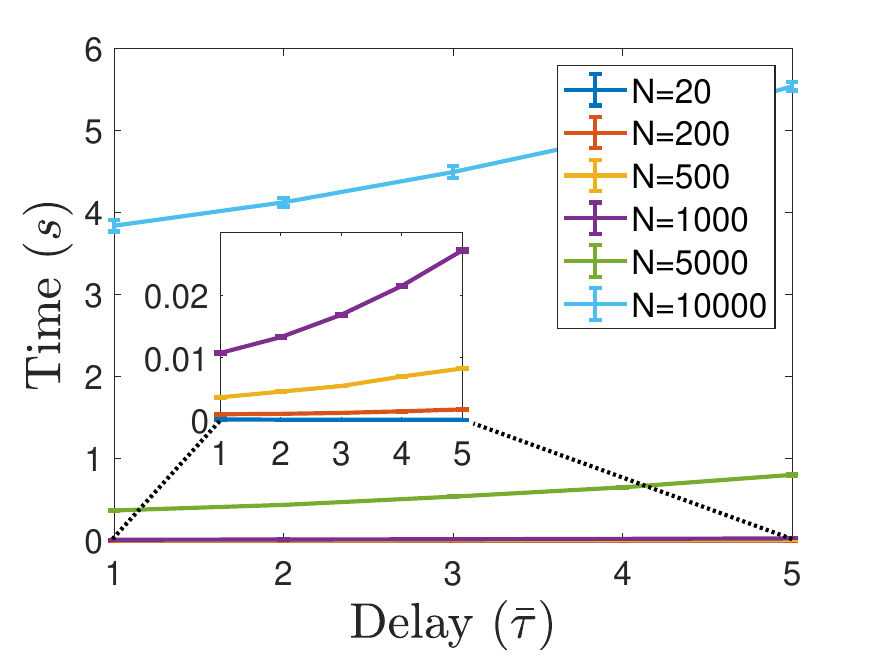}  
        \caption{Iteration time.}
        \label{fig:10k_avg_iter_time}
    \end{subfigure}
    \caption{
    The above figures show the average time to converge (Fig.~\ref{fig:10k_converge_sim_time}) and the average iteration time (Fig.~\ref{fig:10k_avg_iter_time}) in the data centre scale experiments.
    We can see that as we increase the number of nodes, iterations take longer but overall we require less iterations to converge.
    This can be attributed to the low diameter of the graph, which allows more paths of communication between the nodes as their overall count in each network topology increases.
    }
    \label{fig:10k_execution_time}
\end{figure}
As we can see from Fig.~\ref{fig:10k_converge_stats_delay_1} and Fig.~\ref{fig:10k_converge_stats_delay_5} the window size \emph{decreases} as the network size grows. In the presence of low delays (Fig.~\ref{fig:10k_converge_stats_delay_1}) the window is practically zero indicating that the ``min'' and ``max'' converge iteration coincides.
\begin{figure}
    \centering
    \begin{subfigure}{.49\linewidth}
    \centering
    \includegraphics[scale=0.29]{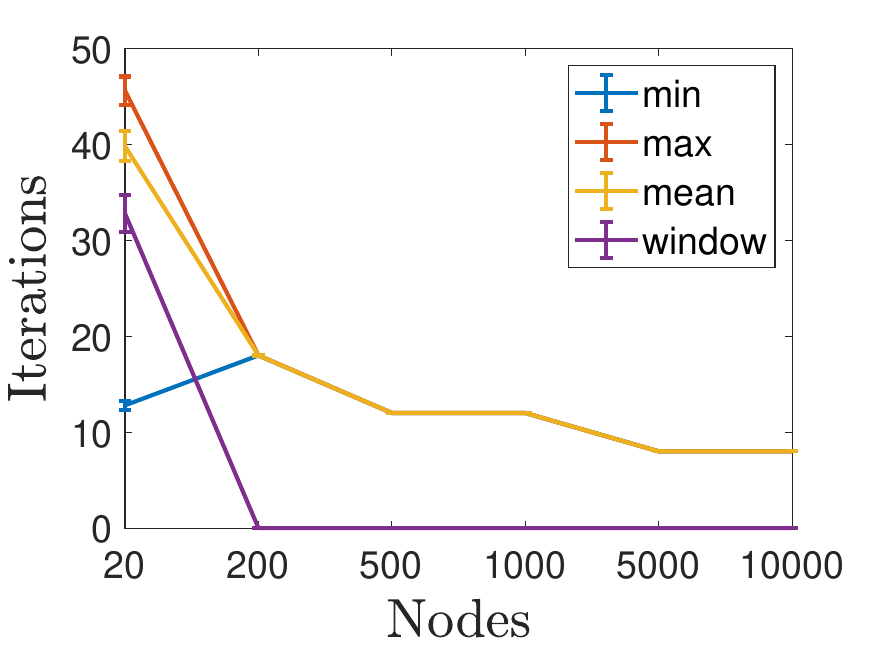}  
    \caption{Delay $\bar{\tau}=1$}
    \label{fig:10k_converge_stats_delay_1}
    \end{subfigure}
    \hfill
    \begin{subfigure}{.49\linewidth}
    \centering
        \includegraphics[scale=0.29]{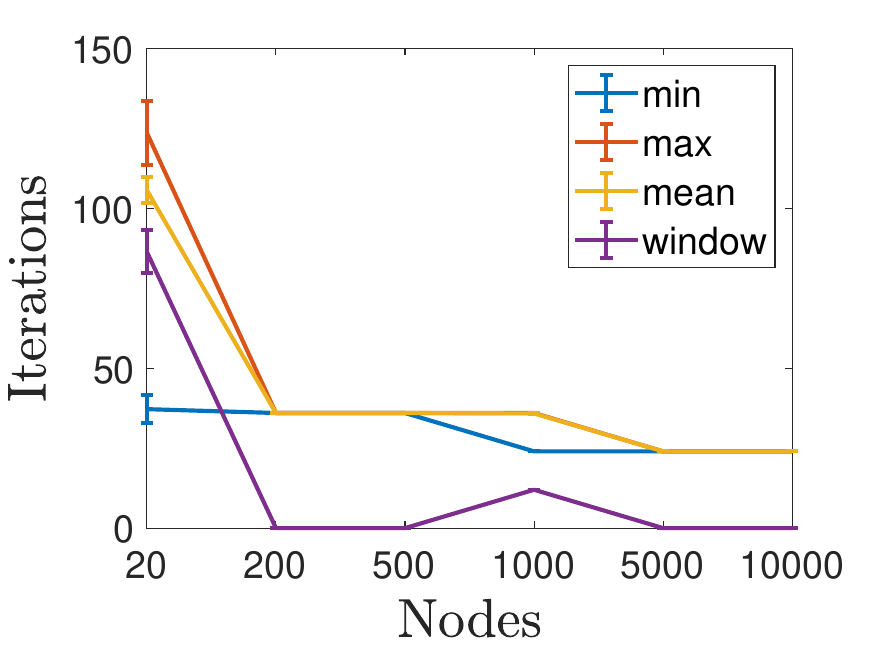}
    \caption{Delay $\bar{\tau}=5$}
    \label{fig:10k_converge_stats_delay_5}
    \end{subfigure}
    \caption{Converge statistics in the presence of both low (Fig.~\ref{fig:10k_converge_stats_delay_1}) and higher upper ((Fig.~\ref{fig:10k_converge_stats_delay_5})) bounds on delays ($\bar{\tau}$), equal to $1$ and $5$ respectively. 
    As the network size grows the window which the first (\emph{min}) and the last (\emph{max}) node converges becomes zero. 
    This indicates that as the network size grows we require fewer iterations to converge and all nodes will converge at the same iteration.
    Note, that as the delay scales delay results in an increase, on average, by a factor of $\approx 2$x to the number of iterations required to converge.
    However, it is worth pointing out that the window to converge remains still very low, and sometimes zero as network size increases.
    }
    \label{fig:10k_converge_stats_combined}
\end{figure}
\noindent Practically speaking, this indicates that the converge variability is low in large networks and is expected to converge in few iterations. 
This means that tasks can be scheduled in a timely fashion and with optimal placement for the given set of jobs. 
This is highly important for any modern data center scheduler aiming to schedule thousands of jobs at-a-time on thousands of nodes in a timely fashion.

\section{Discussion}\label{sec:discussion}
In this paper, we proposed a finite-time asynchronous algorithm for distributively computing a value which a network of nodes can use to make local control decisions. 
Contrary to prior work, our approach is able to operate asynchronously and, as a consequence, also able to handle delays by construction.
To our knowledge this is the first proposed algorithm able to provide finite-time guarantees in the combined delay tolerant and asynchronous setting.

The proposed scheme uses the industry standard CPU utilization model and is able to balance the workload allocation such that each node is allocated tasks proportional to its capabilities.
Concretely, this model defines that the utilization of each CPU core is measured in the bounded range of $[0, 100]$ and indicates the utilization percentage for each individual core within a specified machine~\cite{vmware_performance_2018}.
This effectively allows us to evenly distribute to load across all of the available network nodes loading to better overall cluster utilization.
Practically speaking, it is standard practice in data centers to share the load across the available nodes.
We emphasize that our algorithm algorithm is able to handle both regular workloads as well as bursty ones.
This is achieved because our optimization algorithm works in buckets; where each bucket is filled with incoming jobs.
At the time of scheduling each job within the bucket is attempted to be placed into a suitable node, while guaranteeing the balancing of the overall cluster load.
The time of scheduling is fixed to be at regular intervals or if the bucket is filled.
This behaviour is beneficial for a number of reasons: firstly, bursty workloads are able to be handled gracefully, and secondly, even if there are not enough jobs within the bucket they will still be scheduled in a timely manner.
Note, that our experiments are designed to reflect practical data center deployments which implies that the network graphs considered will be of low diameter and have good connectivity.

Interestingly, as per~\Cref{Algorithm_finitetimeRCwDelays} and a corollary of~\Cref{thm:finitercdelays} the convergence rate is only bounded by the network diameter and its maximum delay.
More importantly, our particular setting implies that packet loss is assumed to be minimal in such deployments but not \textit{delays}. The delays can be attributed to processing and communication delays.
Experiencing processing delays is common in data centers and in the presence of over-provisioned or straggler nodes. 
Communication delays are mainly because of re-transmissions due to packet losses. However, packet losses are not so common and, for this reason, we do not consider them in this work. 
Nevertheless, in case one wishes to consider packet losses as well, this can be achieved by establishing probabilistic guarantees for convergence based on the packet loss distribution.
However, that is beyond the context of this work and is left for future work.

We note that our scheme is asynchronous but in order to successfully operate it implies that the internal clocks of all nodes are paced similarly.
This requirement is necessitated as each node needs to be able to recognize when the appropriate iterations have elapsed.
As noted previously these checks happen every $(1+\hat{\tau})D$ iterations. 
Consistent pacing of each node's clock ensures that the check for convergence at each node will happen at roughly the same time~\cite{lamport_time_2019}.
However, this does not imply that we actually need to synchronize each of the nodes' time-zones nor their actual clocks but, rather, their internal clocks must have similar pacing~\cite{nystrom_uefi_2011}.
This is especially the case after the introduction of High Precision Event Timers in the low-level firmware of most commodity computers and servers alike~\cite{ridoux_case_2011, orosz_performance_2011}.
Notably, this is common practice and present in most modern computers as the clock pacing specification is defined within the Advanced Configuration and Power Interface (ACPI) specifications~\cite{uefi_forum_advanced_nodate,zimmer_beyond_2017}.
In fact, to address these issues in a standardized way hardware manufacturers created the Unified Extensible Firmware Interface (UEFI) forum~\cite{noauthor_unified_nodate}, which is responsible for defining the characteristics and functionality regarding the most basic, low-level functions that each modern computer or server should support. 

As aforementioned in Remark~\ref{remark:2}, a similar approach was proposed in~\cite{prakash_distributed_2020} in the context of average consensus with bounded time-varying delays. Apart from the differences in the application and the fact that we consider asynchronous operation of the nodes, the approach is similar. However, for proving convergence of their proposed algorithm they claim a form of monotonicity of the maximum and minimum values of the states. Specifically, it is claimed \cite[Lemma 3.2]{prakash_distributed_2020} that if the value held by an agent $v_i$ at the present instant of time is strictly lesser (greater) than the maximum (minimum) over the current and delayed values over a horizon $\bar{\tau}$ of all the nodal states, then, the value of agent $v_i$ continues to be strictly lesser (greater) than this maximum (minimum) for all future instants.
Notably, we found several examples of networks for which that statement is not valid.
Practical examples of networks that exhibit such violations are presented in Figures~\ref{fig:violation_20_nodes_diameter_5} and~\ref{fig:violation_50_nodes_diameter_4}.
\begin{figure}[h]
    \centering
    \includegraphics[width=\linewidth]{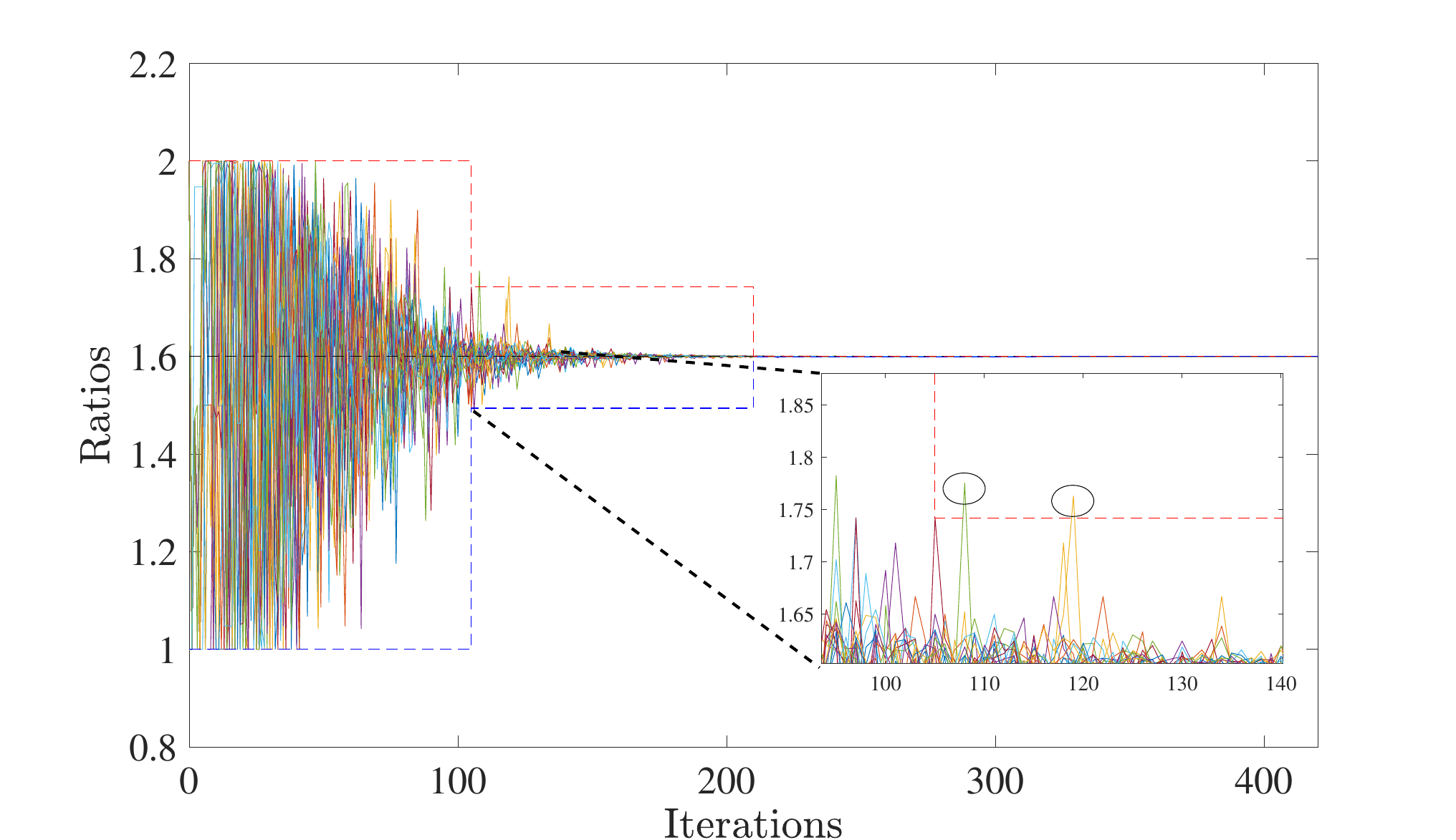}
    \caption{
        Violation in a network consisting of $20$ nodes with a diameter equal to $D=5$ when using a delay of $\tau=20$.
        {Indicatively,} circles indicate violations of {the claim in~\cite[Lemma 3.2]{prakash_distributed_2020}.}
    }
    \label{fig:violation_20_nodes_diameter_5}
\end{figure}
\noindent Concretely, in~\cref{fig:violation_20_nodes_diameter_5} we present a violation that happens in a network comprising of $20$ nodes with a diameter $D=5$ and a delay $\tau=20$.
Interestingly, as we can observe in~\cref{fig:violation_50_nodes_diameter_4} this violation is also observed when dealing with larger networks.
In this particular example presented below the issue is manifested in a network of $50$ nodes with a diameter of $D=4$ and a delay of $\tau=20$.

\begin{figure}[h]
    \centering
    \includegraphics[width=\linewidth]{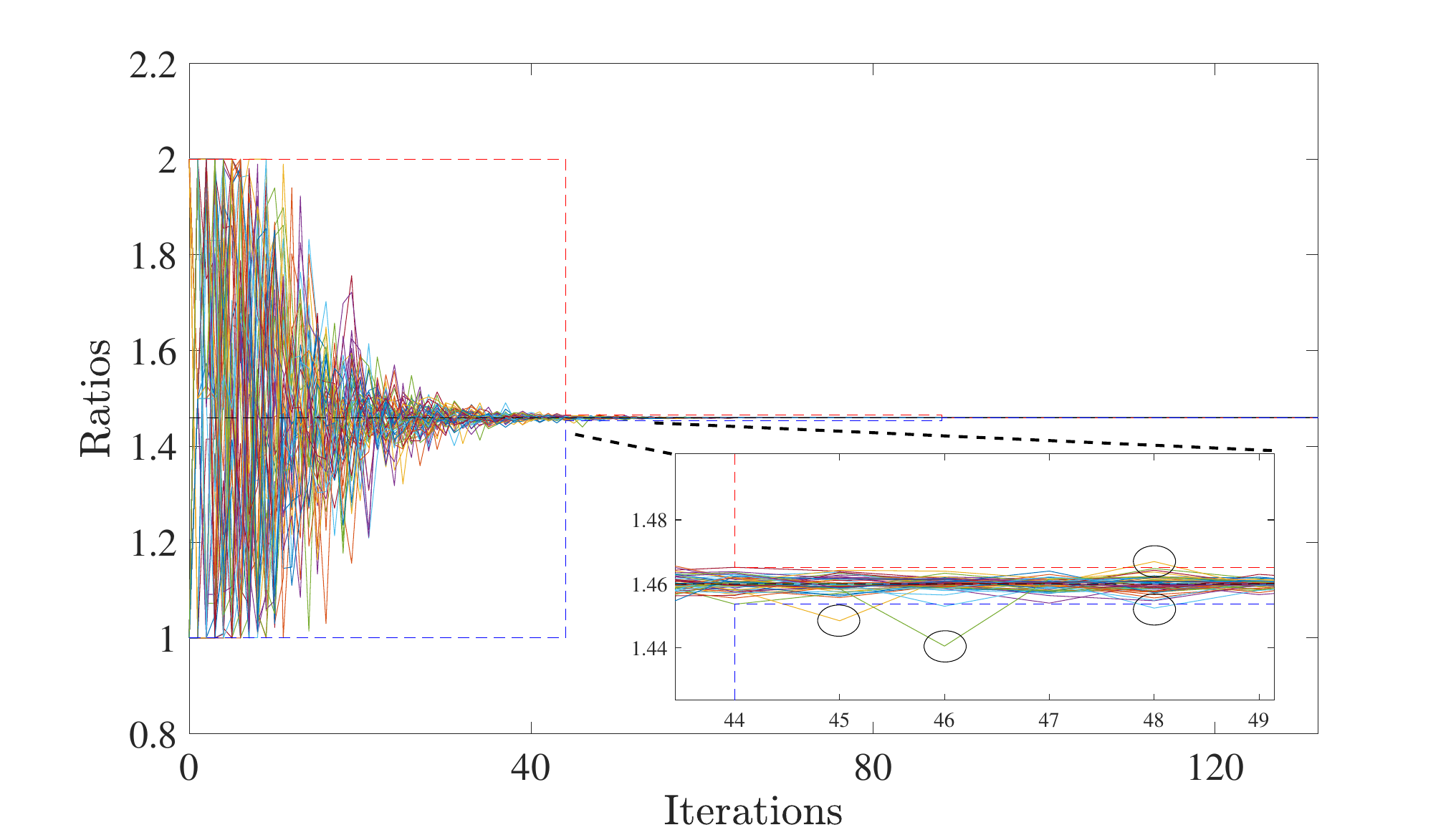}
    \caption{
        Another example of a violation using a larger network consisting of $50$ nodes with a diameter equal to $D=4$ when using a delay of $\tau=10$.
        As in the previous figure, circles indicate violations of the the claim in~\cite[Lemma 3.2]{prakash_distributed_2020}.
    }
    \label{fig:violation_50_nodes_diameter_4}
\end{figure}

Throughout our experiments we observed this behaviour to be more frequent with medium sized networks that had delays greater than $\tau=5$.
On the other hand, the diameter seems to be not a major contributing factor; at least for the values considered in our experiments (e.g., $D$ between $1$ and $10$).

Our solution is able to gracefully handle this situation and still converge into the optimal solution.
The effectiveness of our asynchronous finite-time algorithm was demonstrated on \texttt{CPU} resource allocation in data centers, which can result in better overall system utilization. %
However, one important aspect of such approaches, including our own, is the way they compare against more complex optimization problems. 
In particular against ones that do not have a closed form solution and require complex solvers to be approximated such as ADMM~\cite{jiang_asynchronous_2021}.
As formulated, our problem is able to tackle placement of jobs using the most commonly used CPU utilization model in practical deployments.  %
Furthermore, due to its problem formulation the problem admits a closed-form solution. 
This enables our method to reach the optimization objective significantly faster when compared to more sophisticated solvers such as ADMM; 
especially as the network sizes scale~\cite{chang_asynchronous_2016-1}.
Other approaches have been proposed as well for the same problem formulation~\cite{2021:Rikos}, but the termination of the optimization cannot be synchronized and re-initiating the optimization with the new requests is not possible.
More importantly, we note that our proposed method could also be exploited across multiple domains where asynchronous distributed coordination is desirable (e.g., distributed frequency regulation in microgrids, decentralized computation networks, and voltage control in distribution systems).

\section{Conclusions and Future Directions}\label{sec:conclusions}

\subsection{Conclusions}

In this paper, we proposed a finite-time asynchronous algorithm for distributively computing a value which a network of nodes can use to make local control decisions. 
Contrary to previously-proposed algorithms, our approach works also asynchronously. 
We evaluated our proposed solution using networks of varying delays and diameters which reflected practical data center installations as per common deployment guidelines.
The effectiveness of our asynchronous finite-time algorithm was evaluated against the \texttt{CPU} resource allocation in data centers. 
In turn, more efficient allocation of resources can lead to better overall system responsiveness and utilization. %

\subsection{Future Directions}

Our work can be easily extended to more general convex optimization problems, using gradient-consensus methods, as in~\cite{khatana_gradient-consensus_2020}, but our solution will allow for asynchronous operation and will be able to tolerate delays.

Part of ongoing research focuses on considering deadline constraints and cases for which the workloads exceed the available resources and exploit the heterogeneity of resource units available (e.g. CPU, GPU, and/or accelerators).
In such instances a more sophisticated rejection policy can take place based on inferred resource demands, task priorities, or introduce partial scheduling plans based on either priorities or further, more complex, constraints.

\bibliographystyle{IEEEtran}
\bibliography{main}

\end{document}